\documentclass[sigconf]{acmart}

\renewcommand\footnotetextcopyrightpermission[1]{}
\usepackage{graphicx}
\usepackage{colortbl}
\usepackage{xcolor}
\usepackage[T1]{fontenc}

\usepackage{bm}
\usepackage{amsmath}
\usepackage{adjustbox}
\usepackage{natbib}
\usepackage{amsfonts}
\usepackage{algpseudocode}
\usepackage[nameinlink,capitalize]{cleveref}
\usepackage{float}
\usepackage{algorithm}
\usepackage{subcaption}
\usepackage{caption}
\usepackage{dsfont}
\usepackage{multirow}
\usepackage{array}
\usepackage{enumitem}
\usepackage{nccmath}

\setcopyright{acmlicensed}
\copyrightyear{2024}
\acmYear{2024}
\acmDOI{}

\acmConference[KDD'24]{ACM SIGKDD Conference on Knowledge Discovery and Data Mining}{August 25--29}{Barcelona, Spain}

\begin{document}

\title{Interpretable Transformer Hawkes Processes: Unveiling Complex Interactions in Social Networks}





\begin{abstract}
Social networks represent complex ecosystems where the interactions between users or groups play a pivotal role in information dissemination, opinion formation, and social interactions.
Effectively harnessing event sequence data within social networks to unearth interactions among users or groups has persistently posed a challenging frontier within the realm of point processes.
Current deep point process models face inherent limitations within the context of social networks, constraining both their interpretability and expressive power. These models encounter challenges in capturing interactions among users or groups and often rely on parameterized extrapolation methods when modeling intensity over non-event intervals, limiting their capacity to capture intricate intensity patterns, particularly beyond observed events.
To address these challenges, this study proposes modifications to Transformer Hawkes processes (THP), leading to the development of interpretable Transformer Hawkes processes (ITHP). ITHP inherits the strengths of THP while aligning with statistical nonlinear Hawkes processes, thereby enhancing its interpretability and providing valuable insights into interactions between users or groups. Additionally, ITHP enhances the flexibility of the intensity function over non-event intervals, making it better suited to capture complex event propagation patterns in social networks.
Experimental results, both on synthetic and real data, demonstrate the effectiveness of ITHP in overcoming the identified limitations. Moreover, they highlight ITHP's applicability in the context of exploring the complex impact of users or groups within social networks. 
Our code is available at \url{https://github.com/waystogetthere/Interpretable-Transformer-Hawkes-Process.git}. 
\end{abstract}
\author{Zizhuo Meng$^{*}$}
\affiliation{%
  \institution{University of Technology Sydney}
  \city{Sydney}
  \country{Australia}
}
\email{Zizhuo.Meng@student.uts.edu.au}

\author{Ke Wan$^{*}$}
\affiliation{%
  \institution{University of Illinois Urbana-Champaign}
  \city{Illinois}
  \country{United States}
}
\email{kewan2@illinois.edu}

\author{Yadong Huang}
\affiliation{%
  \institution{Zhoushan Academy of Marine Data Science}
  \city{Zhoushan}
  \country{China}
}
\email{huang-yd14@tsinghua.org.cn}

\author{Zhidong Li}
\affiliation{%
  \institution{University of Technology Sydney}
  \city{Sydney}
  \country{Australia}
}
\email{Zhidong.Li@uts.edu.au}

\author{Yang Wang}
\affiliation{%
  \institution{University of Technology Sydney}
  \city{Sydney}
  \country{Australia}
}
\email{Yang.Wang@uts.edu.au}

\author{Feng Zhou\textsuperscript{\S}}
\affiliation{%
  \institution{Center for Applied Statistics and School of Statistics, Renmin University of China}
  \institution{Beijing Adv. Innov. Ctr. for Future Blockchain and Privacy Computing}
  \city{Beijing}
  \country{China}
}
\email{feng.zhou@ruc.edu.cn}

\maketitle
\def\thefootnote{*}\footnotetext{Equal contributions.}
\def\thefootnote{\S}\footnotetext{Corresponding author.}
\renewcommand\thefootnote{\arabic{footnote}}

\section{Introduction}
Event sequences are pervasive in social networks~\citep{zhang2022counterfactual,kong2023interval}, including platforms such as Stack Overflow, Amazon, and Taobao. Understanding and mining these event sequences to uncover interactions between different users or groups within social networks is a critical research topic~\citep{zipkin2016point,farajtabar2015coevolve}. This analysis can help identify influential users, user groups, and trending topics, offering practical insights for platform optimization and user engagement strategies~\citep{zhao2015seismic,zhou2013learning}. 
For instance, consider the Stack Overflow platform, where developers ask and answer questions related to programming. Event sequences in this context could consist of events such as question postings, answers, comments, and votes. Analyzing this data can reveal insights into user interactions. 

Temporal point processes (TPP)~\citep{daley2003introduction} play a fundamental role in modeling event sequences. The Poisson process~\citep{daley2007introduction}, a basic temporal point process, assumes that events occur uniformly and independently over time. Besides, the Hawkes process~\citep{hawkes1971spectra} is an extension of the Poisson process that allows for event dependencies. 
While these models have been useful in many scenarios, they may not always capture the complexities present in real-world event sequences, which often exhibit more intricate dependencies and interactions.
Therefore, more sophisticated and flexible models are needed.

With the advancement of deep learning, deep architectures have demonstrated remarkable performance in modeling sequence data. 
For example, models utilizing either vanilla RNN~\citep{du2016recurrent} or long short-term memory (LSTM) networks~\citep{mei2017neural} have exhibited improved likelihood fitting and event prediction compared to earlier parameterized models. 
Moreover, models relying on transformer architectures or self-attention mechanisms~\citep{zuo2020transformer,zhang2020self} have shown even better performance. These deep learning approaches have opened up new possibilities for effectively capturing intricate patterns within event sequences, enhancing the overall predictive accuracy and efficiency in various applications. 

However, current deep point process models still have some inherent limitations, which restrict the interpretability and expressive power. 
Firstly, 
such models are unable to explicitly capture the interactions between different event types. 
Deep point process models often model interactions between event types implicitly, which may hinder their interpretability due to the lack of explicit representation for these interactions. 
Understanding the interactions between different event types is crucial in social networks. 
For example, on Amazon, event sequences encompass a wide range of user activities, which can be considered events of various types, including product searches, purchases, reviews, and recommendations. Analyzing the interactions among these types can yield valuable insights into user-level and product-level interactions, providing Amazon with strategic advantages.
Secondly, most existing deep point process models only perform encoding for the positions where events have occurred. For non-event intervals, intensity functions are modelled using parameterized extrapolation methods. For examples, Eq. (11) in~\citet{du2016recurrent}, Eq. (7) in~\citet{mei2017neural}, and Eq. (6) in~\citet{zuo2020transformer}. This approach introduces a parameterized assumption, which restricts the model's expressive power. 

In order to address the aforementioned issues, 
we propose a novel interpretable TPP model based on Transformer Hawkes processes (THP). The proposed model aligns THP perfectly with the statistical nonlinear Hawkes processes, greatly enhancing the interpretability. Thus, we refer to this enhanced model as \emph{interpretable Transformer Hawkes processes (ITHP)}. In ITHP, the attention mechanism's product of the historical event's key and the subsequent event's query corresponds precisely to a time-varying trigger kernel in the statistical nonlinear Hawkes processes. 
By establishing a clear correspondence with statistical Hawkes processes, ITHP offers valuable insights into the interactions between different event types. This advancement is significant for enhancing the interpretability of THP in social network applications. 
Meanwhile, for the intensity function over non-event intervals, we do not adopt a simple parameterized extrapolation method. Instead, we utilize a ``fully attention mechanism'' to express the conditional intensity function at any position. This improvement increases the flexibility of the intensity function over non-event intervals, consequently elevating the model's expressive power. 
Specifically, our contributions are as follows: 
\begin{itemize}
    \item ITHP explicitly captures interactions between event types, providing insights into interactions and improving model interpretability; 
    \item ITHP's fully attention mechanism for the conditional intensity function over non-event intervals enhances model flexibility, allowing it to capture complex intensity patterns beyond the observed events; 
    \item ITHP is validated with synthetic and real social network data, demonstrating its superior ability to interpret event interactions and outperform alternatives in expressiveness. 
\end{itemize}

\section{Related Work}
Enhancing the expressive power of point process models has long been a challenging endeavor. Currently, mainstream approaches fall into two categories.
The first approach entails the utilization of statistical non-parametric methods to augment their expressive capacity. For instance, methodologies grounded in both frequentist and Bayesian nonparametric paradigms are employed to model the intensity function of point processes~\citep{lewis2011nonparametric,zhou2013learning,lloyd2015variational,donner2018efficient,zhou2020auxiliary,zhou2021efficient,pan2021self}. 
The second significant category is deep point process models. These models harness the capabilities of deep learning architectures to infer the intensity function from data,
including RNNs~\citep{du2016recurrent}, LSTM~\citep{mei2017neural,xiao2017modeling}, Transformers~\citep{zuo2020transformer,zhang2020self,zhang2022temporal}, normalizing flow~\citep{shchur2020fast}, adversarial learning~\citep{xiao2017wasserstein, NEURIPS2022_9d3faa41}, reinforcement learning~\citep{upadhyay2018deep}, deep kernel~\citep{okawa2019deep,zhu2021deep,dong2022spatio}, and intensity-free frameworks~\citep{shchur2019intensity}. These architectural choices empower the modeling of temporal dynamics within event sequences and unveil the underlying patterns. 
However, in contrast to statistical point process models, the enhanced expressive power of deep point process models comes at the cost of losing interpretability, rendering deep point process models akin to ``black-box'' constructs. 
To the best of our knowledge, there has been limited exploration into explicitly capturing interactions between event types and enhancing the interpretability of deep point process models~\citep{zhou2023automatic,wei2023granger}. This paper introduces an innovative attention-based ITHP model, whose intensity function aligns seamlessly with statistical nonlinear Hawkes processes, substantially enhancing the interpretability. Our work serves as a catalyst for advancing the interpretability of deep point process models, greatly promoting their utility in uncovering interactions between different users or groups within social networks.

\section{Preliminary Knowledge}
In this section, we provide some background knowledge on some relevant key concepts. 

\subsection{Hawkes Process}
The multivariate Hawkes process~\citep{hawkes1971spectra} is a widely used temporal point process model for capturing interactions among multiple event types. 
The key feature of the multivariate Hawkes process lies in its conditional intensity function. The conditional intensity function $\lambda_k(t|\mathcal{H}_t)$ for event type $k$ at time $t$ is defined as the instantaneous event rate conditioned on the historical information $\mathcal{H}_t = \{(t_i, k_i) | t_i < t\}$: 
\begin{equation*}
    \label{Eq: Hawkes Process} \lambda_k(t|\mathcal{H}_t)=\mu_k+\sum_{t_i<t}\phi_{k,k_i}(t-t_i), 
\end{equation*}
where $\mu_k$ is the base rate for event type $k$ and $\phi_{k,k_i}(t-t_i)$ is the trigger kernel representing the excitation effect from event $t_i$ with type $k_i$ to $t$ with type $k$. It expresses the expected number of occurrences of event type $k$ at time $t$ given the past history of events. 

The interpretability of Hawkes processes stems from its explicit representation of event dependencies through the trigger kernel. The model allows us to quantify the impact of past events with different event types on the occurrence of a specific event, providing insights into the interactions between event types. As a result, the multivariate Hawkes process serves as a powerful tool in social network applications where understanding the interactions between event types (users or groups) is of utmost importance. 

\subsection{Nonlinear Hawkes Process}
\label{nonlinear_hawkes}
In contrast to the original Hawkes process, which assumes only non-negative trigger kernels (excitatory interactions) between events to avoid generating negative intensities, the nonlinear Hawkes process~\citep{bremaud1996stability} offers a more flexible modeling framework by incorporating both excitatory and inhibitory effects among events. 
In the nonlinear Hawkes process, the conditional intensity function for event type $k$ at time $t$ is defined as: 
\begin{equation*}
    \label{Eq: Non-Linear Parametric Hawkes}
    \lambda_k(t|\mathcal{H}_t)=\sigma\left(\mu_k+ \sum_{t_i<t}\phi_{k,k_i}(t-t_i)\right), 
\end{equation*}
where $\sigma(\cdot)$ is a nonlinear mapping from $\mathbb{R}$ to $\mathbb{R}^+$, ensuring the non-negativity of the intensity. Hence this trigger kernel can be positive (excitatory) or negative (inhibitory), thus enabling the modeling of complex interactions between different event types. 

In the aforementioned models, the trigger kernel depends solely on the \textbf{relative time} $t-t_i$, implying that the trigger kernel is shift-invariant. However, in dynamic Hawkes process models~\citep{zhou2020fast,zhou2020auxiliary,bhaduri2021change,zhou2022efficient}, the trigger kernel is further extended to vary with \textbf{absolute time}, denoted as $\phi(t-t_i,t_i)$. By incorporating the absolute time, the trigger kernel becomes capable of capturing time-varying patterns, offering the model more degrees of freedom in its representation. 

\subsection{Transformer Hawkes Process}
Our work is built upon THP~\citep{zuo2020transformer}, so we concisely introduce the framework of THP here. 
Given a sequence $\mathcal{S} = \{(t_i, k_i)\}_{i=1}^{L}$ where each event is characterized by a timestamp $t_i$ and an event type $k_i$, THP leverages two types of embeddings, namely temporal embedding and event type embedding, to represent these two kinds of information.
To encode event timestamps, THP represents each timestamp $t_i$ using an embedding vector $\mathbf{z}(t_i)\in \mathbb{R}^{M}$:
\begin{equation*}
z_j(t_i)=
\begin{cases}
\cos(t_i/10000^\frac{j-1}{M} ) \text{ if } j \text{ is odd},\\
\sin(t_i/10000^\frac{j}{M} ) \text{ if } j \text{ is even}, 
\end{cases}
\end{equation*}
where $z_j(t_i)$ is the $j$-th entry of $\mathbf{z}(t_i)$ and $j=0,\ldots,M-1$. 
The collection of time embeddings is represented as $\mathbf{Z}=\left[\mathbf{z}(t_1), \ldots, \mathbf{z}(t_L)\right]^\top \in \mathbb{R}^{L \times M}$.
For encoding event types, the model utilizes a learnable matrix $\mathbf{U} \in \mathbb{R}^{M \times K}$, where $K$ is the number of event types. For each event type $k_i$, its embedding $\mathbf{e}(k_i)$ is computed as:
\begin{equation*}
   \mathbf{e}(k_i)=\mathbf{U} \mathbf{y}_i \in \mathbb{R}^M,  
\end{equation*}
where $\mathbf{y}_i$ is the one-hot encoding of the event type $k_i$. The collection of type embeddings is $\mathbf{E}=[\mathbf{e}(k_1), \ldots, \mathbf{e}(k_L)]^\top \in \mathbb{R}^{L \times M}$. 
The final embedding is the summation of the temporal and event type embeddings: 
\begin{equation}
    \mathbf{X}=\mathbf{Z}+\mathbf{E}\in\mathbb{R}^{L \times M}, 
    \label{feature_x}
\end{equation}
where each row of $\mathbf{X}$ represents the complete embedding of a single event in the sequence $\mathcal{S}$. 

After embedding, the model focuses on learning the dependence among events using self-attention mechanism. The attention output $\mathbf{S}$ is computed as: 
\begin{gather*}
\label{eq:transformer attention layer}
\mathbf{S} = \text{softmax}\left(
\frac{{\mathbf{Q} \mathbf{K}}^\top}{\sqrt{M_K}}\right)\mathbf{V} \in \mathbb{R}^{L \times M_V}, \\
\mathbf{Q} = \mathbf{XW}^Q \in \mathbb{R}^{L \times M_K}, \mathbf{K} = \mathbf{XW}^K \in \mathbb{R}^{L \times M_K}, \mathbf{V} = \mathbf{XW}^V \in \mathbb{R}^{L \times M_V},
\end{gather*}
where $\mathbf{Q}$, $\mathbf{K}$, $\mathbf{V}$ are the query, key and value matrices. 
Matrices $\mathbf{W}^Q \in \mathbb{R}^{M \times M_K}$, $\mathbf{W}^K \in \mathbb{R}^{M \times M_K}$ and $\mathbf{W}^V \in \mathbb{R}^{M \times M_V}$ are the learnable parameters. 
To preserve causality and prevent future events from influencing past events, we mask out the entries in the upper triangular region of $\mathbf{QK}^\top$. 

Finally, the attention output $\mathbf{S}$ is passed through a two-layer MLP to produce the hidden state $\mathbf{H}$: 
\begin{equation*}
\label{eq: transformer hidden}
\mathbf{H} = \text{ReLU}(\mathbf{SW}_1 + \mathbf{b}_1)\mathbf{W}_2 + \mathbf{b}_2 \in \mathbb{R}^{L \times M},
\end{equation*}
where $\mathbf{W}_1\in\mathbb{R}^{M \times M_H}$, $\mathbf{W}_2\in\mathbb{R}^{M_H \times M}$, $\mathbf{b}_1\in\mathbb{R}^{M_H}$ and $\mathbf{b}_2\in\mathbb{R}^{M}$ are the learnable parameters. 
The $k$-type conditional intensity function of THP is designed as: 
\begin{equation}
\lambda_k(t|\mathcal{H}_t)=\text{softplus}\left(\colorbox{red!40}{$\alpha_k(t-t_i)/t_i$}+\mathbf{w}_k^\top\mathbf{h}(t_i)+b_k\right), 
    \label{thp_lamda}
\end{equation}
where $t_i$ is the last event before $t$, $\alpha_k,\mathbf{w}_k,b_k$ are learnable parameters, the nonlinear function is chosen to be softplus to ensure that the intensity is non-negative, $\mathbf{h}(t_i)$ is the transpose of the $i$-th row of $\mathbf{H}$ expressing the historical impact on event $t_i$.

\section{Interpretable Transformer Hawkes Processes}
As mentioned earlier, THP has two prominent limitations:
(1) THP implicitly model the dependency between events, which hinders the explicit representation of interactions between different event types and makes it challenging to understand the interactions among event types. (2) Like many other deep point process models, THP applies attention encoding only to the event occurrence positions, while using parameterized extrapolation methods to model the intensity on non-event intervals (the red term in \cref{thp_lamda}). This approach introduces a parameterized assumption restricting the model's expressive power. 





To enhance the model's interpretability and expressiveness, our work introduces modifications to the THP model. \textbf{Specifically, we make modifications to (1) the event embedding, (2) the attention module and (3) the conditional intensity function in THP.}
Interestingly, the modified THP corresponds perfectly to the statistical nonlinear Hawkes processes. 
This leads to significantly improved interpretability and a better characterization of the interactions between event types.
Additionally, new design of the conditional intensity function can avoid the restrictions imposed by parameterized extrapolation, enabling the model to effectively capture complex intensity patterns beyond the observed events. 

In following sections, we outline the step-by-step process of modifying THP to achieve the aforementioned goals. For each modification, we provide theoretical proofs to demonstrate the rationality and validity of the respective changes.


\subsection{Modified Event Embedding}
In ITHP, we maintain the same temporal embedding and event type embedding methods as in THP.
However, our modification lies in replacing the summation operation in \cref{feature_x} with concatenation:
\begin{equation}
\mathbf{X} = \mathbf{Z} + \mathbf{E} \in \mathbb{R}^{L \times M} \Rightarrow \mathbf{X} = [\mathbf{Z}, \mathbf{E}] \in \mathbb{R}^{L \times 2M}.
\label{new_feature_x}
\end{equation}

The reason for this modification is that the original summation operation introduces a similarity between timestamps and event types (or vice versa) of the preceding and succeeding events. 
However, in statistical Hawkes processes, known for their interpretability, the interaction between two events is the magnitude of a kernel determined by the similarity (correlation) between their types and the similarity (distance) between their timestamps. No cross-similarity is introduced.
To maintain a similar level of interpretability, we replace summation with concatenation here. 

\begin{theorem}
In \cref{new_feature_x}, the concatenation operation enables us to explicitly capture the desired temporal and event type similarities, while simultaneously avoiding any cross-similarities between timestamps and event types. 
\end{theorem}
\begin{proof}
Suppose we define $\mathbf{X}$ using concatenation, and in subsequent attention module computation, it is necessary to calculate the product $\mathbf{X}\mathbf{X}^\top$ to measure the similarity between different data points. The similarity between the $i$-th point and the $j$-th point can be expressed as follows:
\begin{equation*}
\mathbf{X}_i\mathbf{X}_j^\top=[\mathbf{Z}_i, \mathbf{E}_i][\mathbf{Z}_j,\mathbf{E}_j]^\top = \mathbf{Z}_i\mathbf{Z}_j^\top+\mathbf{E}_i\mathbf{E}_j^\top. 
\end{equation*}
Instead, if we define $\mathbf{X}$ using addition, $\mathbf{X}_i\mathbf{X}_j^\top$ is as follows: 
\begin{equation*}
(\mathbf{Z}_i+\mathbf{E}_i)(\mathbf{Z}_j+\mathbf{E}_j)^\top= \mathbf{Z}_i\mathbf{Z}_j^\top+\mathbf{E}_i\mathbf{E}_j^\top+\mathbf{Z}_i\mathbf{E}_j^\top+\mathbf{E}_i\mathbf{Z}_j^\top. 
\end{equation*}
It is evident that by defining $\mathbf{X}$ through concatenation, temporal and event type similarities are captured separately. Otherwise, the cross-similarities emerge. 
\end{proof}

\subsection{Modified Attention Module}
In ITHP, we still use self-attention to capture the influences of historical events on subsequent events. However, unlike THP, in the modified attention module, we use distinct query and key matrices: 
\begin{equation}
\begin{gathered}
\label{new_transformer_attention}
\mathbf{S} = \text{softmax}\left(
\frac{{\mathbf{Q} \mathbf{K}}^\top}{\sqrt{M_K}}\right)\mathbf{V} \in \mathbb{R}^{L \times M_V}, \\
\mathbf{Q} = \mathbf{XW}^Q \in \mathbb{R}^{L \times M_K} \Rightarrow \mathbf{Q} = \mathbf{X} \in \mathbb{R}^{L \times 2M}, \\
\mathbf{K} = \mathbf{XW}^K \in \mathbb{R}^{L \times M_K} \Rightarrow \mathbf{K} = \mathbf{X} \in \mathbb{R}^{L \times 2M}, \\
\mathbf{V} = \mathbf{XW}^V \in \mathbb{R}^{L \times M_V}, 
\end{gathered}
\end{equation}
where the $i$-th row of $\mathbf{S}$, $\mathbf{S}_i$, represents the historical influence on the $i$-th event. The calculation of $\mathbf{S}_i$ can be explicitly expressed as the summation over all events preceding event $i$, where the attention weights are normalized by the softmax: 
\begin{equation}
    \mathbf{S}_i = \sum_{j<i}\underset{j<i}{\text{softmax}}\left(\frac{\mathbf{X}_i\mathbf{X}_j^\top}{\sqrt{2M}}\right)\mathbf{V}_j \in \mathbb{R}^{M_V}. 
\label{new_S}
\end{equation}

The reason for this modification is that after removing $\mathbf{W}^Q$ and $\mathbf{W}^K$, the attention weights can be simply represented as $\mathbf{XX}^\top$. Compared to the original $\mathbf{QK}^\top$, $\mathbf{XX}^\top$ has a clearer physical meaning. In $\mathbf{XX}^\top$, the entry in the $i$-th row and $j$-th column can be expressed as a shift-invariant function $g_{k_i,k_j}(t_i-t_j)$. 
This representation allows for a more meaningful interpretation of the relationship between events. In contrast, $\mathbf{QK}^\top$ does not achieve this clarity. 
\begin{theorem}
\label{theorem32}
Assuming that $\mathbf{X}$ is obtained through the concatenation operation in \cref{new_feature_x}, after omitting $\mathbf{W}^Q$ and $\mathbf{W}^K$ in \cref{new_transformer_attention}, the entry at the $i$-th row and $j$-th column of $\mathbf{XX}^\top$ can be expressed as a shift-invariant function $g_{k_i,k_j}(t_i-t_j)$, where $t_j < t_i$.
\label{theorem_S}
\end{theorem}
\begin{proof}
When we use \cref{new_feature_x} to obtain $\mathbf{X}$ and remove $\mathbf{W}^Q$ and $\mathbf{W}^K$ in \cref{new_transformer_attention}, the similarity between the $i$-th and $j$-th points, denoted as $\mathbf{X}_i\mathbf{X}_j^\top$, is expressed as $\mathbf{X}_i\mathbf{X}_j^\top = \mathbf{Z}_i\mathbf{Z}_j^\top + \mathbf{E}_i\mathbf{E}_j^\top$. 
If we assume that $M$ is even, $\mathbf{Z}_i\mathbf{Z}_j^\top$ can be further represented as:
\begin{align*}
\mathbf{Z}_i\mathbf{Z}_j^\top &= \sum_{m=1}^{\frac{M}{2}+1}\cos (t_i\omega_m)\cos(t_j\omega_m) + \sin (t_i\omega_m)\sin(t_j\omega_m)\\
&= \sum_{m=1}^{\frac{M}{2}+1}\cos((t_i - t_j)\omega_m),
\end{align*}
where $\omega_m = 1/10000^{2(m-1)/M}$.
It is clear that $\mathbf{X}_i\mathbf{X}_j^\top$ can be expressed as a shift-invariant function $g_{k_i,k_j}(t_i-t_j)$, with $t_i-t_j$ originating from $\mathbf{Z}_i\mathbf{Z}_j^\top$, and the subscripts $k_i, k_j$ arising from $\mathbf{E}_i\mathbf{E}_j^\top$. While retaining $\mathbf{W}^Q$ and $\mathbf{W}^K$ leads to the following expression:
\begin{equation*}
\mathbf{Q}_i\mathbf{K}_j^\top 
= [\mathbf{Z}_i, \mathbf{E}_i]\mathbf{W}^Q{\mathbf{W}^K}^\top[\mathbf{Z}_j,\mathbf{E}_j]^\top,
\end{equation*}
where the introduction of $\mathbf{W}^Q{\mathbf{W}^K}^\top$ can once again introduce undesired cross-similarities and render the temporal similarity term $\mathbf{Z}_i[\mathbf{W}^Q{\mathbf{W}^K}^\top]_{\mathbf{Z}_i\mathbf{Z}_j}\mathbf{Z}_j^\top$ unable to be expressed in a shift-invariant function form. 
\end{proof}
\begin{corollary}
Given \cref{theorem_S}, we can further simplify \cref{new_S} as follows: 
\begin{equation}
    \mathbf{S}_i = \sum_{j<i}\mathbf{g}^\top_{k_i,k_j}(t_i-t_j,t_j) \in \mathbb{R}^{M_V}, 
\label{corollary_s}
\end{equation}
where $\mathbf{g}$ is an $M_V$ dimensional vector function. 
\end{corollary}
\begin{proof}
According to \cref{theorem32}, $\mathbf{X}_i\mathbf{X}_j^\top$ can be expressed as a shift-invariant function $g_{k_i,k_j}(t_i-t_j)$. After normalization through softmax in \cref{new_S}, we obtain $\tilde{g}_{k_i,k_j}(t_i-t_j)$ satisfying $\sum_{j<i}\tilde{g}_{k_i,k_j}(t_i-t_j)=1$. When multiplied by $\mathbf{V}_j$, $\tilde{g}_{k_i,k_j}(t_i-t_j)\mathbf{V}_j$ yields a vectorized, time-varying and non-normalized function $\mathbf{g}^\top_{k_i,k_j}(t_i-t_j,t_j)$ where the additional $t_j$ stems from the introduction of $\mathbf{V}_j$. 
\end{proof}

\subsection{Modified Conditional Intensity Function}
The form of \cref{corollary_s} naturally reminds us of the trigger kernel summation in statistical Hawkes processes. The only difference is that $\mathbf{g}$ in \cref{corollary_s} is a vector, whereas the trigger kernel in statistical Hawkes processes is a scalar function. Taking inspiration from this, we propose a more interpretable conditional intensity function:
\begin{equation}
\begin{aligned}
&\lambda_k(t|\mathcal{H}_t)=\text{softplus}\left(\sum_{t_i<t}\colorbox{yellow!40}{$\mathbf{g}^\top_{k,k_i}(t-t_i,t_i)\mathbf{w}_k$}
+\colorbox{green!40}{$b_k$}\right) \\
&={\text{softplus}}\left(\sum_{t_i<t}
\colorbox{yellow!40}{$\underset{t_i<t}{\text{softmax}}\left(\frac{\mathbf{X}_t\mathbf{X}_i^\top}{\sqrt{2M}}\right)\mathbf{V}_i\mathbf{w}_k$}
+\colorbox{green!40}{$b_k$}\right), 
\end{aligned}
\label{new_intensity}
\end{equation}
where $\mathbf{w}_k$ is a learnable parameter used to aggregate the vector $\mathbf{g}$ into a scalar value and $b_k$ is a learnable bias term. 
The newly designed conditional intensity aligns perfectly with the nonlinear Hawkes processes with a time-varying trigger kernel. 
\textbf{The green term $b_k$ in \cref{new_intensity} corresponds to the base rate $\mu_k$ in \cref{nonlinear_hawkes}, the yellow term $\mathbf{g}^\top_{k,k_i}(t-t_i,t_i)\mathbf{w}_k$ in \cref{new_intensity} corresponds to the time-varying trigger kernel $\phi_{k,k_i}(t-t_i,t_i)$ in \cref{nonlinear_hawkes}, and the softplus function serves as a non-linear mapping ensuring the non-negativity of the intensity.} This leads to improved interpretability as the trigger kernel can be explicitly expressed in our design, in contrast to the original THP. 

\subsection{Fully Attention-based Intensity Function}
\label{section: Fully Attention-based Intensity Function}
In point process model training with maximum likelihood estimation (MLE), it is vital to compute the intensity integral over the entire time domain, which requires modeling the intensity both at event positions and on non-event intervals. 
In the RNN-based deep point process models~\citep{du2016recurrent,mei2017neural}, due to the limitations of the RNN framework in solely modeling latent representations at event positions, the aforementioned works adopted parameterized extrapolation methods to model the intensity on non-event intervals, see Eq. (11) in \cite{du2016recurrent} and Eq. (7) in \cite{mei2017neural}. THP \cite{zuo2020transformer} also adopted the same approach to model the intensity on non-event intervals (the red term in \cref{thp_lamda}). 
However, we emphasize that \textbf{attention-based deep point process models do not necessarily require the parameterized extrapolation methods to model the intensity on non-event intervals}. 
Our design \cref{new_intensity} employs the attention mechanism to model the intensity function whether it is at event positions or not. Therefore, we refer to it as a ``fully attention-based intensity function''. The fully attention-based intensity function circumvents the limitations of parameterization and ensures that the model can effectively capture intricate intensity patterns at non-event positions, thus enhancing the model's expressive power. 



\subsection{Model Training}
For a given sequence $\mathcal{S} = \{(t_i, k_i)\}_{i=1}^{L}$ on $[0, T]$, the point process model training can be performed by the MLE approach. The log-likelihood of a point process is expressed in the following form: 
\begin{equation}
\mathcal{L}(\mathcal{S})=\sum_{i=1}^L\log\lambda_{k_i}(t_i|\mathcal{H}_{t_i})-\int_0^T\lambda(t|\mathcal{H}_t)dt, 
\label{loglikelihood}
\end{equation}
where $\lambda(t|\mathcal{H}_t)=\sum_{k=1}^K\lambda_k(t|\mathcal{H}_t)$.

For ITHP, we estimate its parameters by maximizing the log-likelihood. Regarding the first term, we only need to compute the intensity function at event positions using \cref{new_intensity}. As for the second term, the intensity integral generally lacks an analytical expression. Here, we employ numerical integration by discretizing the time axis into a sufficiently fine grid and calculating the intensity function at each grid point using \cref{new_intensity}. 

\textbf{Complexity:}
The utilization of a fine grid does not significantly increase computational time. This is because the attention mechanism facilitates parallel computation of attention outputs for each point. This parallelized computation 
improves the scalability of ITHP. 
Parallel computation with more grid points would require additional memory. Fortunately, for one-dimensional temporal point processes, a large number of grids is not necessary. In subsequent experiments, all datasets can run smoothly with only 8GB memory. 

\section{Experiment}
We assess the performance of ITHP using both synthetic and public datasets. With the synthetic dataset, our objective is to validate the interpretability of our model by accurately identifying the underlying ground-truth trigger kernel. For the public datasets, we conduct a comprehensive evaluation of ITHP by comparing its performance against popular baseline models. 
The goal here is twofold: to quantitatively demonstrate the superior expressive power of ITHP and to qualitatively analyze its interpretability on real datasets. 
\subsection{Synthetic Data}
\label{section: Toy Data}
We validate the interpretability of ITHP using two sets of 2-variate Hawkes processes data. Each dataset is simulated from a 2-variate Hawkes processes described in~\cref{Eq: Hawkes Process}, using the thinning algorithm~\citep{ogata1998space}. Both datasets share a common base rate ($\mu = 0.2$), but they possess distinct trigger kernels: 
\begin{itemize}
    \item \textbf{Exponential Decay Kernel} This kernel assumes that the influence of historical events decays exponentially as time elapses. The kernel function is given by: $\phi_{ij}(\tau) = \alpha_{ij} \exp(-\beta_{ij} \tau)$ for $\tau>0$, where $j$ is the source type and $i$ is the target type. Specifically,  $\alpha_{11}=\alpha_{22}=3$, $\alpha_{12}=2$, $\alpha_{21}=1$ and $\beta_{ij} = 5$ for all $i,j$. 
    \item \textbf{Half Sinusoidal Kernel} This kernel assumes the influence of historical events follows a sinusoidal pattern as time elapses and disappears when the interval surpasses $\pi$. The kernel function is given by: $\phi_{ij}(\tau) =\alpha_{ij} \sin(\tau)$ for $0< \tau < \pi$. Likewise, $j$ is the source type and $i$ is the target type. Specifically, $\alpha_{11}=\alpha_{22}=0.33$, $\alpha_{12}=0.1$, and $\alpha_{21}=0.05$. 
\end{itemize}
Further elaboration on the simulation process and statistical aspects of the synthetic dataset can be found in \cref{Appendix: Synthetic Data}. 
\begin{figure*}[tb]
    \begin{center}
    \adjustbox{valign=b}{
    \begin{minipage}{0.24\textwidth}
    \includegraphics[width=\columnwidth, height=0.75\columnwidth]{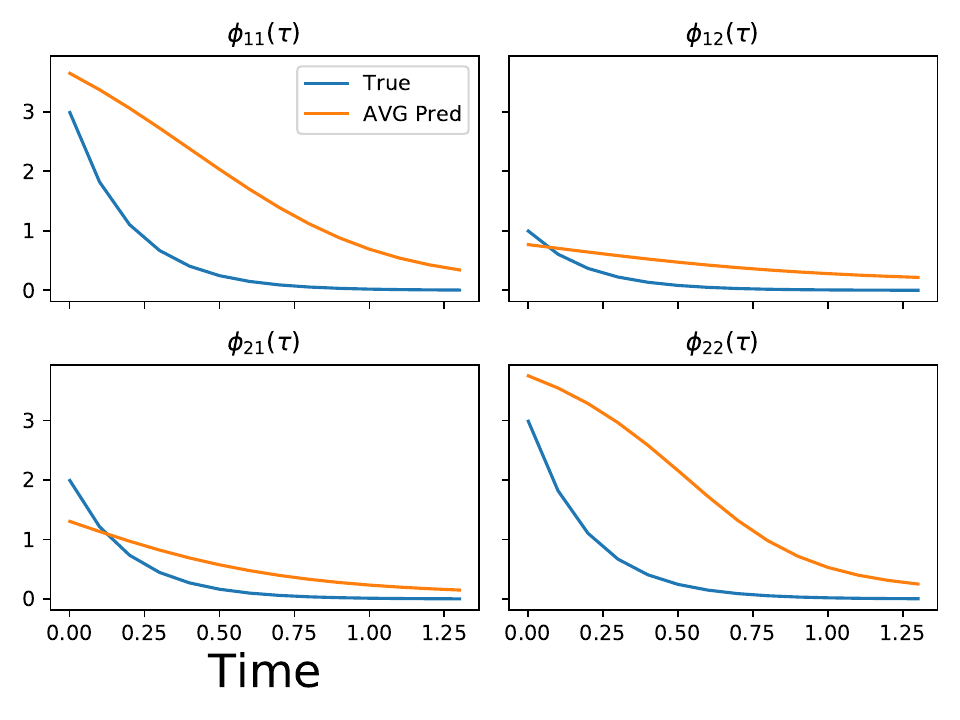}
    \subcaption[]{Trigger Kernel Recover (Exp)}
    \label{fig: exp kernel recover}
    \end{minipage}}
    \adjustbox{valign=b}{
    \begin{minipage}{0.24\textwidth}
    \includegraphics[width=\columnwidth, height=0.75\columnwidth]{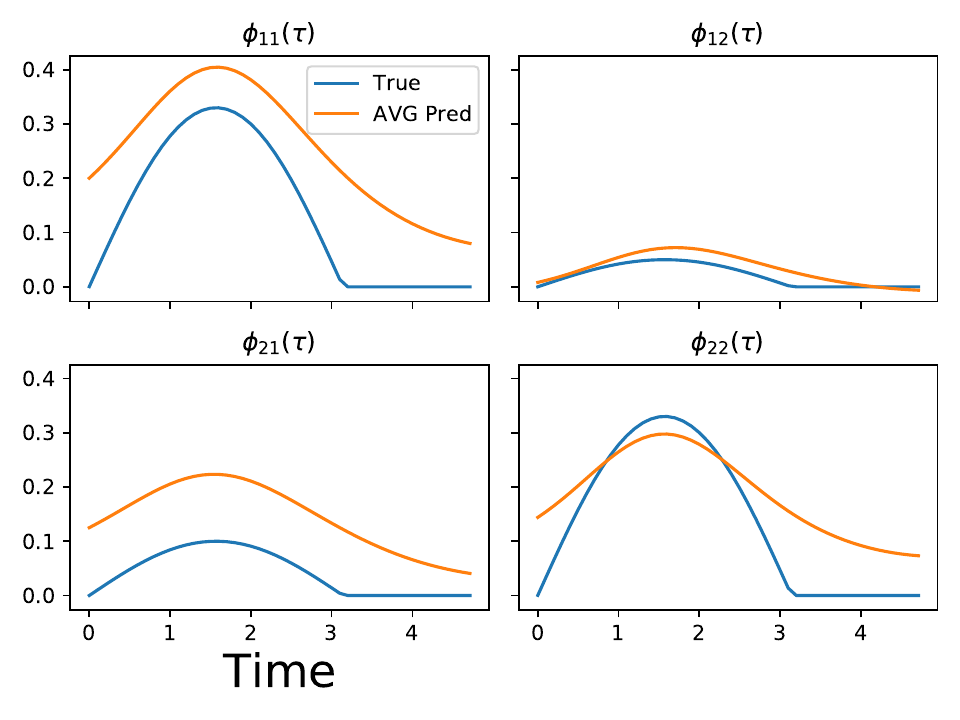}
    \subcaption[]{Trigger Kernel Recover (Sin)}
    \label{fig: half-sin kernel recover}
    \end{minipage}}
    \adjustbox{valign=b}{
    \begin{minipage}{0.24\textwidth}
    \includegraphics[width=\columnwidth, height=0.75\columnwidth]{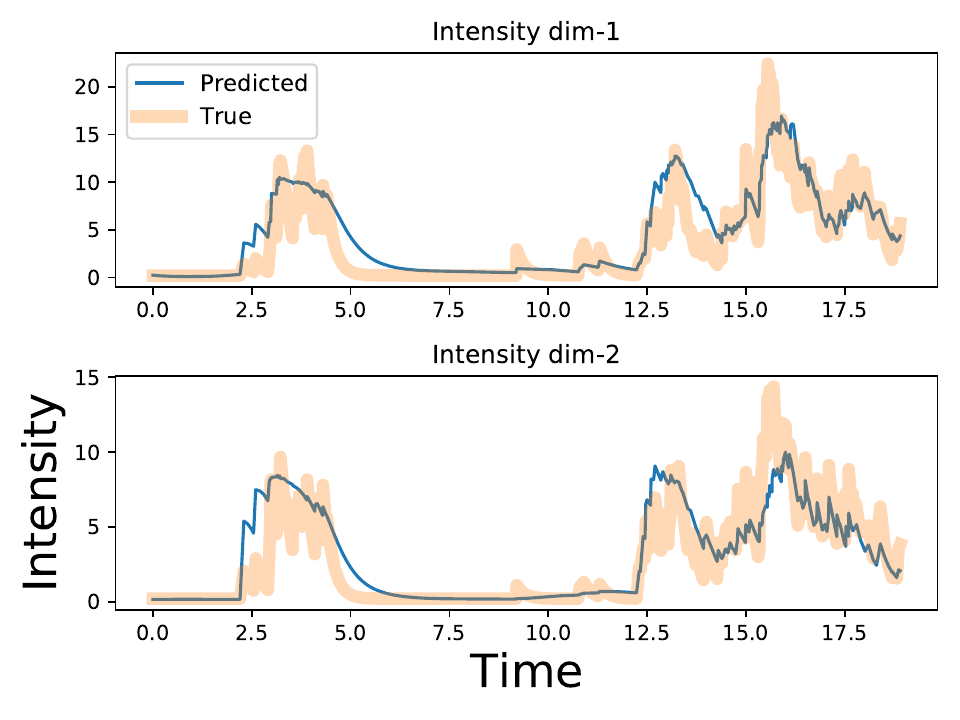}
    \subcaption[]{Intensity Recover (Exp)}
    \label{fig: Intensity Recover(Sin)}
    \end{minipage}}
    \adjustbox{valign=b}{
    \begin{minipage}{0.24\textwidth}
    \includegraphics[width=\columnwidth, height=0.75\columnwidth]{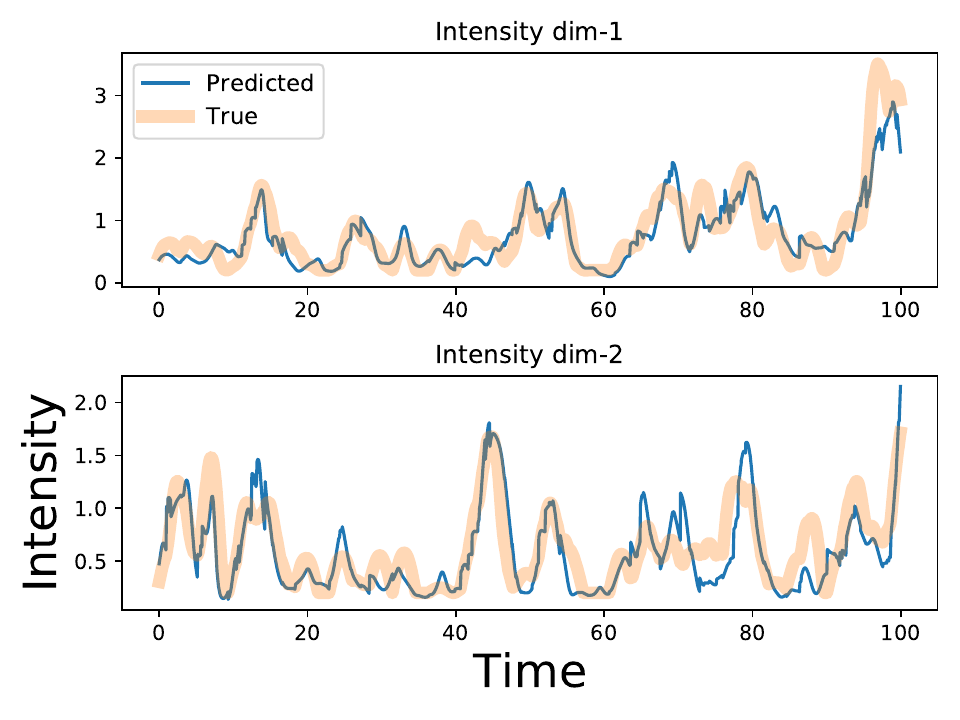}
    \subcaption[]{Intensity Recover (Sin)}
    \label{fig: Intensity Recover(Exp)}
    \end{minipage}}


    \end{center}
    \caption{Experimental results with synthetic data. (a)(b) Comparison between the ground-truth trigger kernel and the one learned by ITHP: (a) exponential decay and (b) half sinusoidal. (c)(d) Comparison between the ground-truth intensity and the one learned by ITHP: (c) exponential decay and (d) half sinusoidal.}
    
    \label{fig: Synthetic Experiment Result}
\end{figure*}

\textbf{Results:}
We validate the interpretability of ITHP by reconstructing the trigger kernel. In ITHP, the trigger kernel is represented as $\mathbf{g}^\top_{k,k_i}(t - t_i, t_i)\mathbf{w}_k$, which is time-varying. To uncover the time-invariant trigger kernel inherent in the synthetic dataset, we evaluate trigger kernels at various time points and compute their mean. This approach enables us to extract the desired time-invariant trigger kernel~\citep{zhang2020self}. The results are presented in \cref{fig: exp kernel recover,fig: half-sin kernel recover}, revealing a noticeable alignment between the learned kernel trends and the patterns exhibited by the ground-truth kernels. 
Moreover, as depicted in \cref{fig: Intensity Recover(Exp),fig: Intensity Recover(Sin)}, the learned intensity function from ITHP exhibits a striking resemblance to the ground-truth intensity function. This observation underscores ITHP's capability to accurately capture the true conditional intensity function for both exponential decay and half sinusoidal Hawkes processes.
\begin{figure}

    \centering
    \includegraphics[width=0.65\columnwidth]{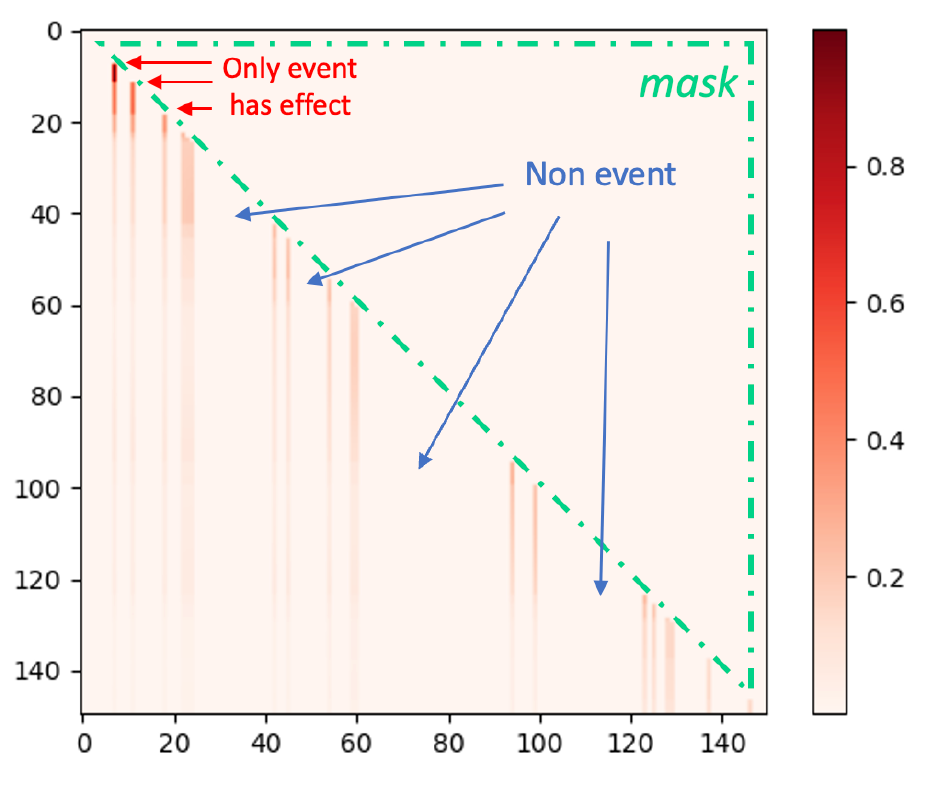}
    \vspace*{-3mm}
    \caption{The attention weight matrix for events and grids in the case of exponential decay kernel. 
    Horizontal axis: source point, Vertical axis: target point. 
    It is evident that events have an impact on the future which decays over time. Grids within non-event intervals do not exert any influence as they are not actual events.}
    \label{fig: Attention Map (Exp)}
    \vspace*{-5mm}
\end{figure}
We also visualize the learned attention map of ITHP, which provides a deeper insight into the influence patterns. As depicted in \cref{fig: Attention Map (Exp)}, this is the attention weight matrix of a testing sequence in the context of exponential decay Hawkes process data. The sequence encompasses both event timestamps and grids within the non-event intervals.
In the matrix, the rows and columns correspond to events and grids on the sequence (arranged chronologically). The horizontal axis represents the source point, while the vertical axis represents the target point. 
Only events have the potential to impact subsequent points, whereas grids, lacking actual event occurrences, cannot affect future points. As a result, it is evident that numerous columns corresponding to grids have values of $0$. Due to a masking operation, the upper triangular section, including the diagonal, is set to $0$, which restricts events from influencing the past.
Moreover, the color of event columns becomes progressively lighter as time advances, which aligns with the characteristics of the ground-truth exponential decay trigger kernel. 

\subsection{Public Data}
In this section, we extensively evaluate ITHP by comparing it to baseline models across several public datasets. 
We have selected several network-sequence datasets, including social media (StackOverflow), online shopping (Amazon, Taobao), traffic networks (Taxi), and a widely used public synthetic dataset (Conttime). 

\subsubsection{Datasets}
We investigate five public datasets, each accompanied by a concise description. 
More details can be found in \cref{Appendix: Public Data}. 
\begin{itemize}
    \item
    \textbf{StackOverflow}\footnote{\url{https://snap.stanford.edu/data/}}\citep{snapnets}: This dataset has two years of user awards on a question-answering website: StackOverflow. Each user received a sequence of badges (Nice Question, Good Answer, $\ldots$) and there are $K = 22$ kinds of badges. 
    \item
    \textbf{Amazon}\footnote{\url{https://nijianmo.github.io/amazon/}}\citep{ni2019justifying}: This dataset includes user online shopping behavior events on Amazon website (browsing, purchasing, $\dots$) and there are in total $K=16$ event types. 
    \item 
    \textbf{Taobao}\footnote{\url{https://tianchi.aliyun.com/dataset/649}}\citep{zhu2018learning}: This dataset is released for the 2018 Tianchi Big Data Competition and comprises user activities on Taobao website (browsing, purchasing, $\dots$) and there are in total $K=17$ event types. 
    \item
    
    \textbf{Taxi}\footnote{\url{https://chriswhong.com/open-data/foil_nyc_taxi/}}\citep{whong2014foiling}: While our main focus is social networks, our model can also be applied to other domains. This dataset comprises traffic-network sequences, including taxi pick-up and drop-off incidents across five boroughs of New York City. Each borough, whether involved in a pick-up or drop-off event, represents an event type and there are in total $K = 2\times5=10$ event types. 

    \item 
    \textbf{Conttime}\citep{mei2017neural}: This dataset is a popular public synthetic dataset designed for Hawkes processes, which comprises ten thousand event sequences with event types $K=5$.

\end{itemize}

\subsubsection{Baselines}
In the experiments, we conduct a comparative analysis against the following popular baseline models: 
\begin{itemize}
    \item \textbf{RMTPP}~\citep{du2016recurrent} is a RNN-based model. It learns the representation of influences from historical events and takes event intervals as input explicitly. 
    \item \textbf{NHP}~\citep{mei2017neural} utilizes a continuous-time LSTM network, which incorporates intensity decay, allowing for a more natural representation of temporal dynamics without requiring explicit encoding of event intervals as inputs to the LSTM. 
    \item \textbf{SAHP}~\citep{zhang2020self} uses self-attention to characterize the influence of historical events and enhance its predictive capabilities by capturing intricate dependencies within the data. 
    \item \textbf{THP}~\citep{zuo2020transformer} is another attention-based model that utilizes Transformer to capture event dependencies while maintaining computational efficiency. 
\end{itemize}

\begin{table*}[th]
\centering
\caption{The TLL and ACC of ITHP and other four baselines on five public datasets. Champion is in bold, runner-up is underlined. 
Note that the Ex-ITHP is the extrapolated ITHP, which is a modified version of ITHP used for ablation study. More details are provided in \cref{section: ablation study}.}
\label{table: Real Data Experiment}
\begin{sc}
\scalebox{1}{
\begin{tabular}{ccccccccccc}
    \toprule
    \multirow{2}{*}{Model} & \multicolumn{2}{c}{stackoverflow} & \multicolumn{2}{c}{amazon}  & \multicolumn{2}{c}{taobao}  & \multicolumn{2}{c}{taxi} & \multicolumn{2}{c}{conttime}\\
    \cmidrule{2-11}
        & TLL($\uparrow$) & ACC($\uparrow$)   & TLL($\uparrow$) & ACC($\uparrow$) & TLL($\uparrow$) & ACC($\uparrow$) & TLL($\uparrow$) & ACC($\uparrow$) & TLL($\uparrow$) & ACC($\uparrow$) \\
    \midrule
    RMTPP & $-2.87_{\pm 0.02}$ & $0.43_{\pm 0.01}$ & $-2.68_{\pm 0.03}$ & $0.30_{\pm 0.01}$ & $-3.81_{\pm 0.05}$ & $0.44_{\pm 0.03}$ & $0.17_{\pm 0.04}$ & $0.91_{\pm 0.01}$ & $-1.88_{\pm 0.03}$ & $0.38_{\pm 0.01}$\\
    NHP & $-2.80_{\pm 0.01}$ & $0.43_{\pm 0.02}$ & $-2.70_{\pm 0.05}$ & $0.27_{\pm 0.01}$ & $\underline{-3.10}_{\pm 0.02}$ & $0.45_{\pm 0.01}$ & $\underline{0.24}_{\pm 0.04}$ & $\underline{0.93}_{\pm 0.04}$ & $\underline{-1.54}_{\pm 0.01}$ & $\underline{0.41}_{\pm 0.03}$\\
    SAHP & $\mathbf{-1.96}_{\pm 0.02}$ & $\underline{0.45}_{\pm 0.01}$ & $\mathbf{-1.42}_{\pm 0.04}$ & $\underline{0.35}_{\pm 0.01}$ & $-4.70_{\pm 0.03}$ & $\underline{0.46}_{\pm 0.01}$ & $0.21_{\pm 0.03}$ & $\mathbf{0.94}_{\pm 0.01}$ & $-2.22_{\pm 0.02}$ & $\mathbf{0.42}_{\pm 0.01}$\\
    THP & $-3.41_{\pm 0.01}$ & $\mathbf{0.46}_{\pm 0.01}$ & $-3.26_{\pm 0.21}$ & $0.34_{\pm 0.01}$ & $-4.76_{\pm 0.11}$ & $0.44_{\pm 0.05}$ & $0.22_{\pm 0.05}$ & $\underline{0.93}_{\pm 0.02}$ & $-3.16_{\pm 0.19}$ & $0.34_{\pm 0.01}$\\
    ITHP & $\underline{-2.50}_{\pm 0.03}$ & $\mathbf{0.46}_{\pm 0.01}$ & $\underline{-2.10}_{\pm 0.02}$ & $\mathbf{0.36}_{\pm 0.01}$ & $\mathbf{-3.09}_{\pm 0.02}$ & $\mathbf{0.47}_{\pm 0.01}$ & $\mathbf{0.25}_{\pm 0.05}$ & $\mathbf{0.94}_{\pm 0.01}$ & $\mathbf{-1.43}_{\pm 0.01}$ & $0.38_{\pm 0.01}$\\
    \midrule
    \midrule
    Ex-ITHP & $-3.58_{\pm 0.01}$ &$0.43_{\pm 0.03}$	&$-4.65_{\pm 0.02}$ &$0.33_{\pm 0.01}$	&$-4.80_{\pm 0.01}$ &$0.41_{\pm 0.02}$& $0.18_{\pm 0.03}$ &$0.85_{\pm 0.03}$& $-3.74_{\pm 0.04}$ & $0.31_{\pm 0.02}$  \\
    \bottomrule
\end{tabular}}
\end{sc}
\end{table*}

\subsubsection{Metrics}
We assess ITHP and other baseline models using two distinct metrics:
\begin{itemize}
    \item \textbf{TLL}: the log-likelihood on the test data which quantifies the model's ability to capture the underlying data distribution and effectively predict future events. 
    \item \textbf{ACC}: the event type prediction accuracy on the test data which characterizes the model's accuracy in predicting the specific types of events, thereby gauging its capacity to discriminate between different event categories.
\end{itemize}

\subsubsection{Quantitative Analysis}
We conduct a comparative experiment across five datasets using all baseline models. The results, as shown in \cref{table: Real Data Experiment}, demonstrate that ITHP can achieve competitive performance. A more intuitive visualization is presented in \cref{fig: TLL Comparison}, where each model's TLL is standardized by subtracting the TLL of ITHP.
It is worth noting that, to achieve interpretability, ITHP undergoes a certain degree of simplification, resulting in a reduction of its number of parameters. Interestingly, we observe that ITHP achieves comparable performance to other models with larger number of parameters. ITHP can equivalently be regarded as a non-parametric, time-varying, and nonlinear statistical Hawkes process. The results in \cref{fig: TLL Comparison} provide some reflections: while deep point processes claim to outperform statistical point processes, it is evident that a sufficiently flexible (non-parametric, time-varying, and nonlinear) statistical point process can also achieve competitive performance. 
Furthermore, ITHP maintains excellent interpretability, both at the event level and the event type level. \cref{fig: Attention Map(SO)} displays the attention weight matrix of a testing sequence from StackOverflow, illustrating the impact between events: the influence from past events tends to decrease as time elapsed.
Moreover, ITHP can describe the influence functions between event types. Take StackOverflow as an example: \cref{fig: Learned Kernel Demo} presents the learned influence functions from types {1,3,4,5,9,12} to type 4, which is the most prevalent type. Generally, these influences tend to decay over time.

\begin{figure*}[th]
    \begin{center}
    \adjustbox{valign=b}{
    \begin{minipage}{0.24\textwidth}
    \includegraphics[width=\columnwidth]{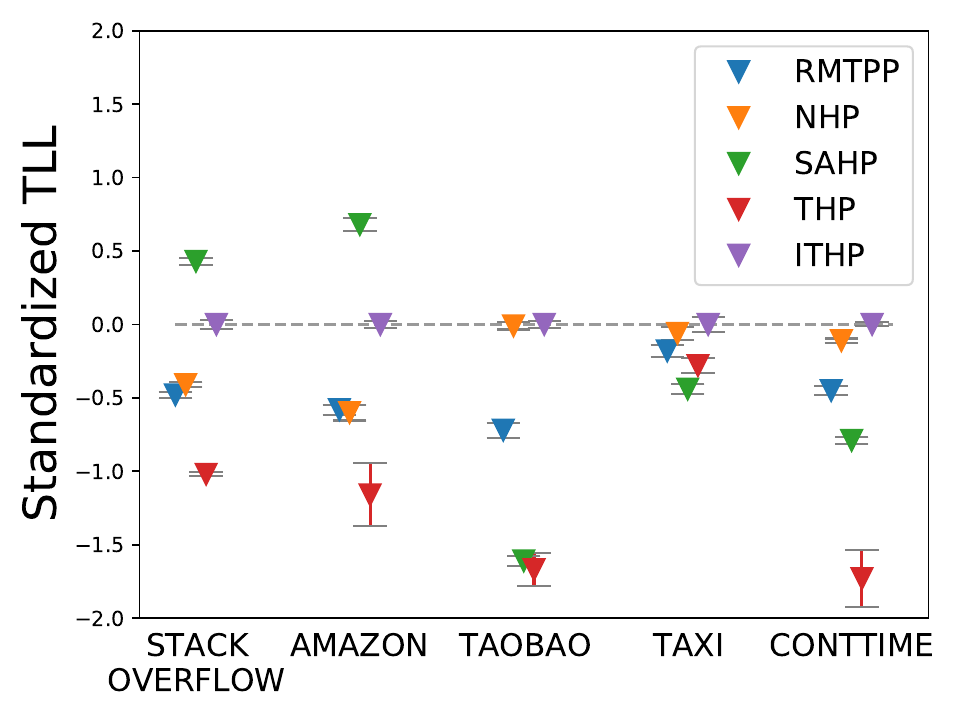}
    \subcaption[]{TLL Comparison}
    \label{fig: TLL Comparison}
    \end{minipage}}
    \adjustbox{valign=b}{
    \begin{minipage}{0.24\textwidth}
    \includegraphics[width=\columnwidth, height=0.75\columnwidth]{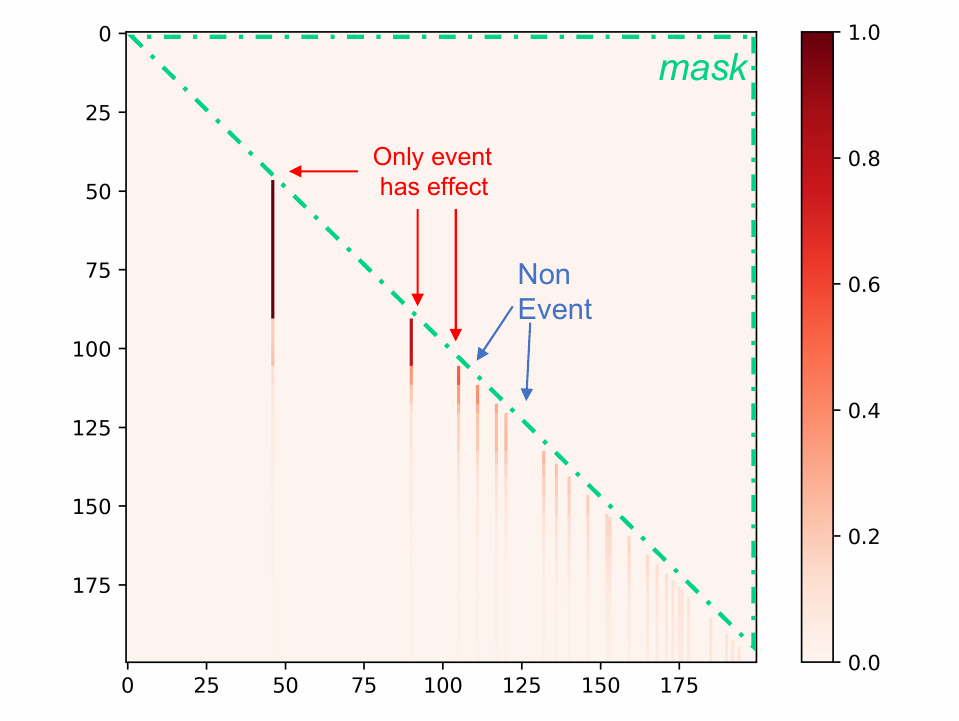}
    \subcaption[]{Attention Map (StackOverflow)}
    \label{fig: Attention Map(SO)}
    \end{minipage}}
    \adjustbox{valign=b}{
    \begin{minipage}{0.24\textwidth}
    \includegraphics[width=\columnwidth]{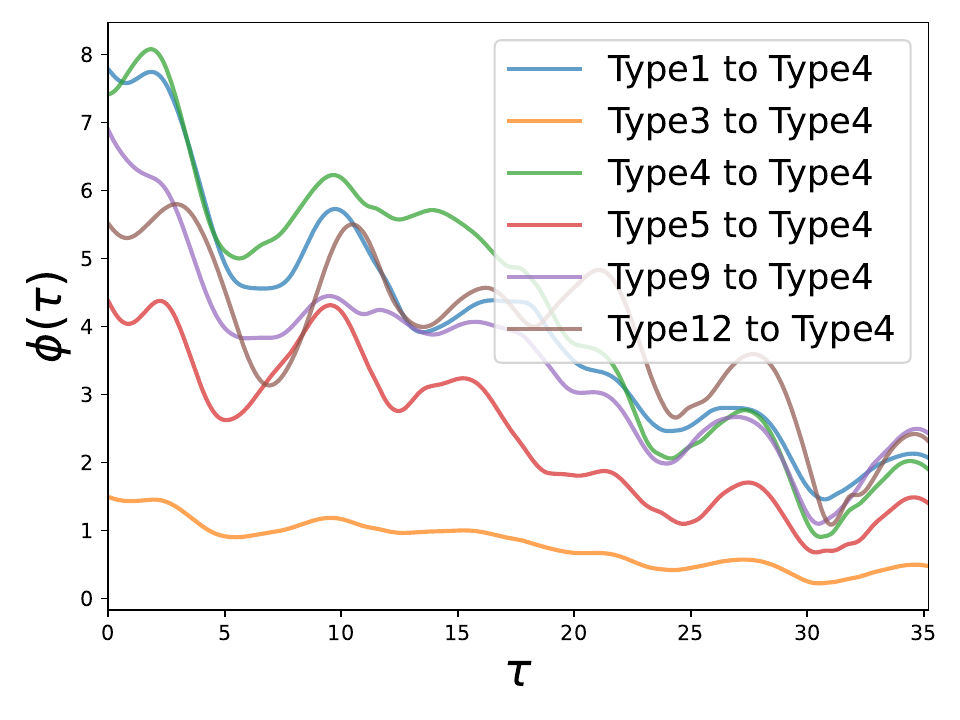}
    \subcaption[]{Estimated $\hat{\phi}(\tau)$ (StackOverflow)}
    \label{fig: Learned Kernel Demo}
    \end{minipage}}
    \adjustbox{valign=b}{
    \begin{minipage}{0.24\textwidth}
    \includegraphics[width=\columnwidth]{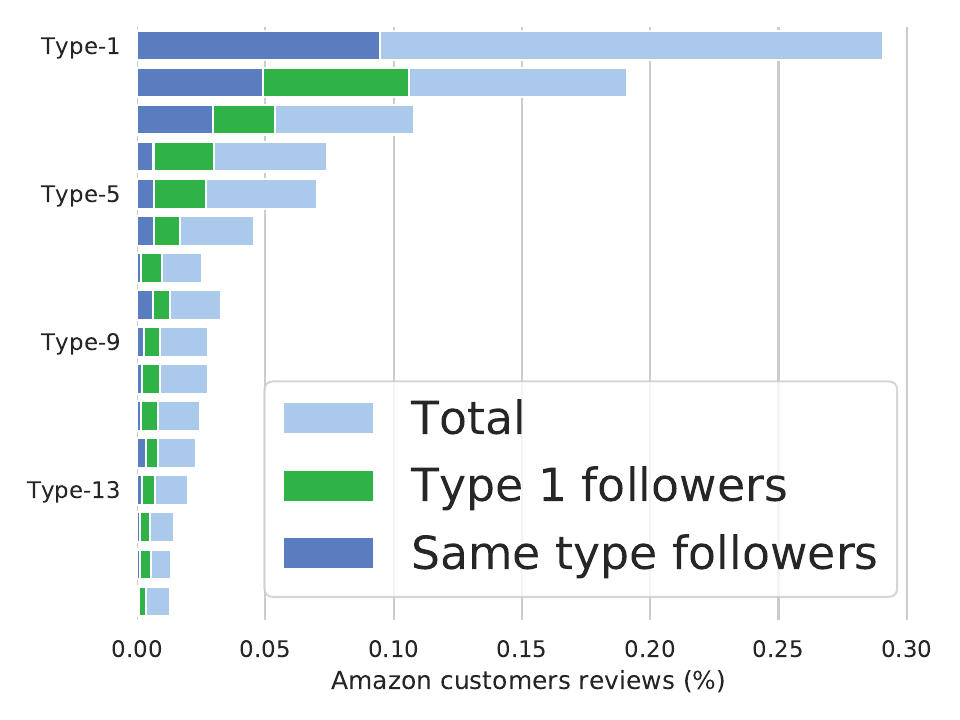}
    \subcaption[]{Amazon Statistics}
    \label{fig: Amazon statistics}
    \end{minipage}}
    \\
    \adjustbox{valign=b}{
    \begin{minipage}{0.24\textwidth}
    \includegraphics[width=\columnwidth]{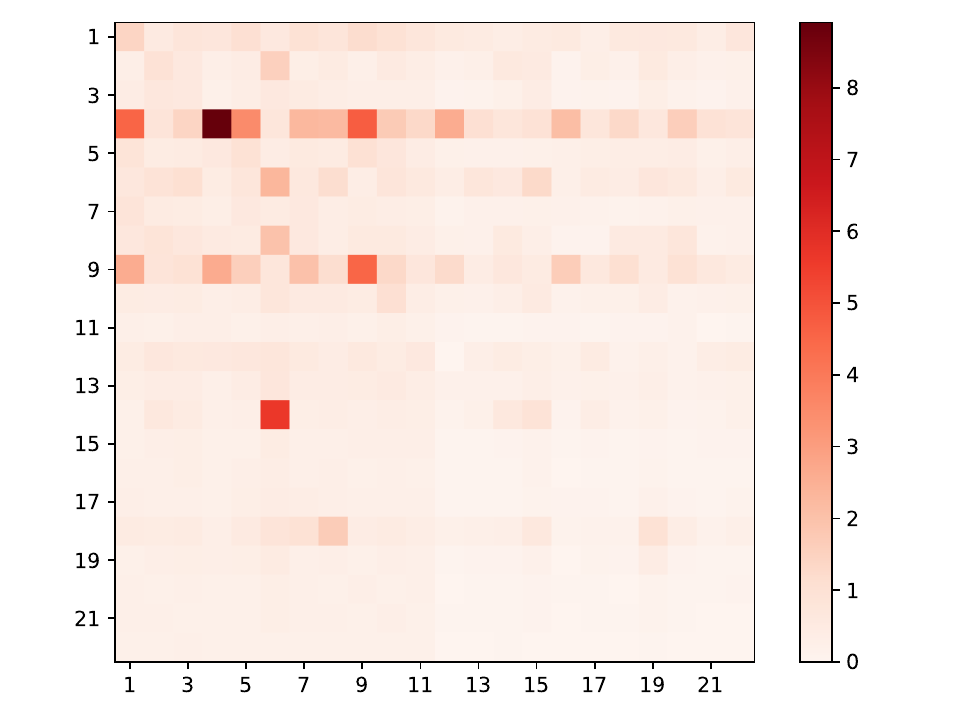}
    \subcaption[]{Heatmap of StackOverflow}
    \label{fig: HeatMap(SO)}
    \end{minipage}}
    \adjustbox{valign=b}{
    \begin{minipage}{0.24\textwidth}
    \includegraphics[width=\columnwidth]{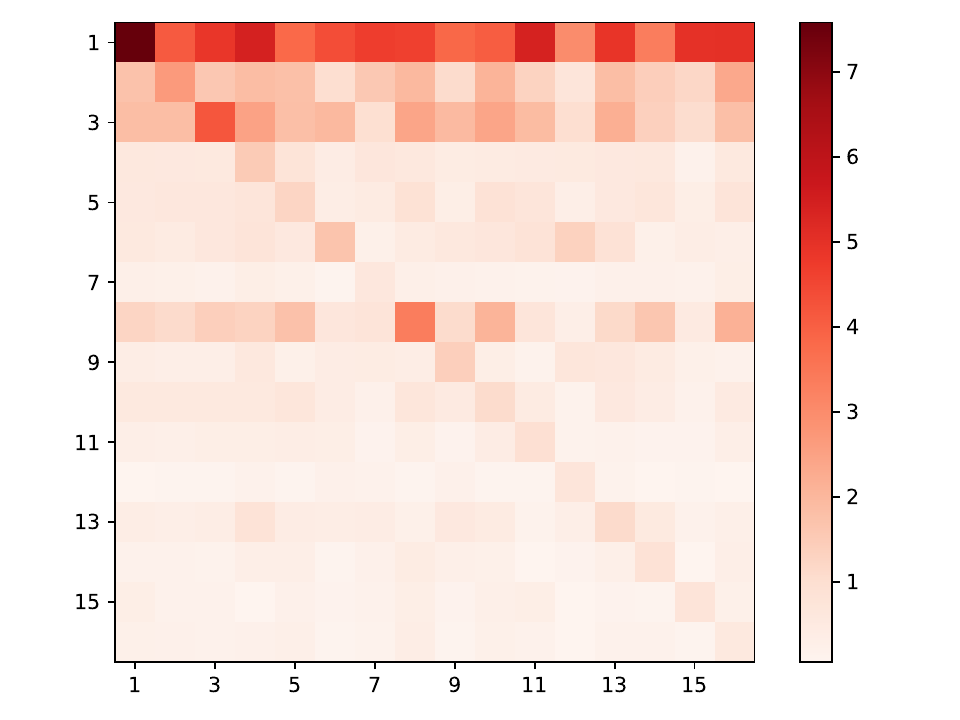}
    \subcaption[]{Heatmap of Amazon}
    \label{fig: HeatMap(Amazon)}
    \end{minipage}}
    \adjustbox{valign=b}{
    \begin{minipage}{0.24\textwidth}
    \includegraphics[width=\columnwidth]{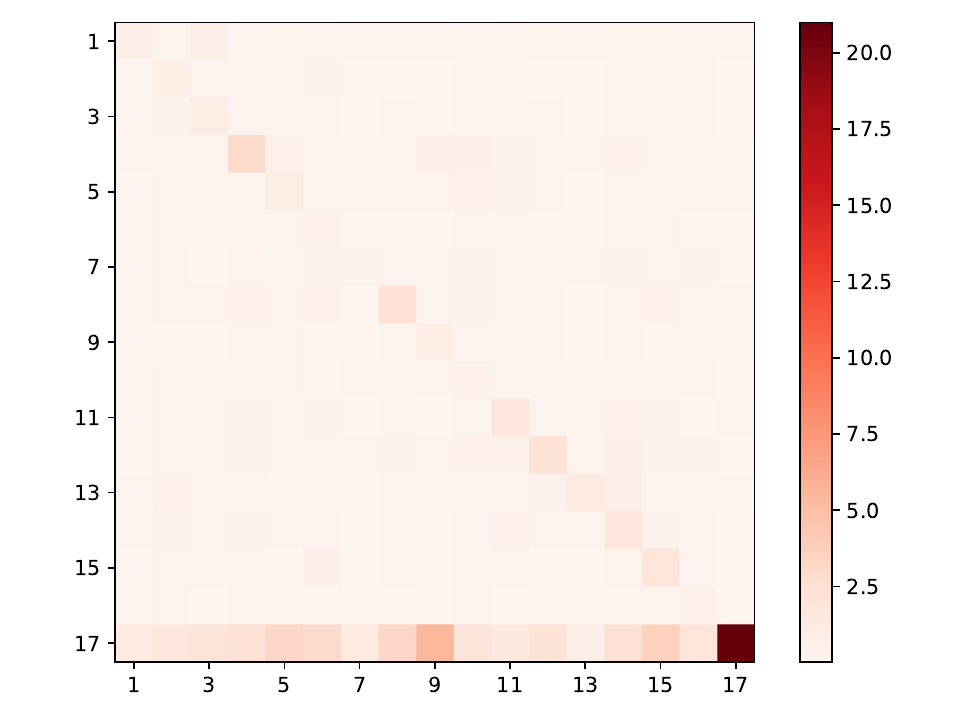}
    \subcaption[]{Heatmap of Taobao}
    \label{fig: HeatMap(Taobao)}
    \end{minipage}}
    \adjustbox{valign=b}{
    \begin{minipage}{0.24\textwidth}
    \includegraphics[width=\columnwidth]{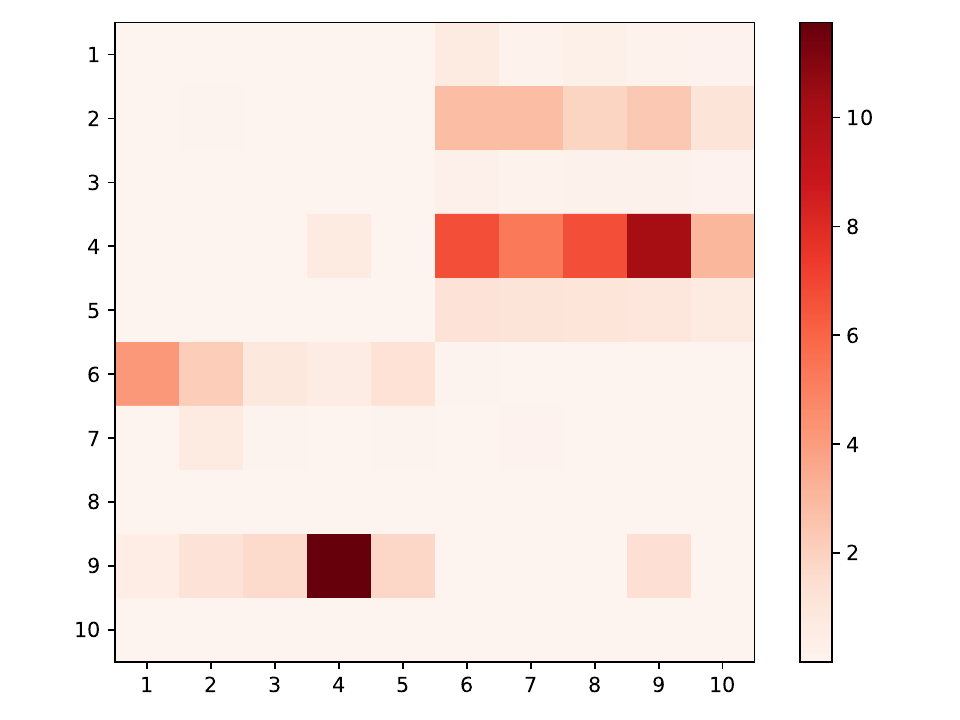}
    \subcaption[]{Heatmap of Taxi}
    \label{fig: HeatMap(Taxi)}
    \end{minipage}}
    \end{center}
    \caption{Experimental results with public data. (a) Comparative analysis of standardized TLL across models on five public datasets.
(b) Attention weight matrix for the first 200 events and grids in a StackOverflow sequence. 
Similar patterns to synthetic data were observed.
(c) Learned influence function $\hat{\phi}(\tau)$ among event types in StackOverflow, with influences generally decaying over time. 
(d) A statistical analysis on Amazon: the percentages of various event types (``Total''), the percentages of the next event being of the same type (``Same type follower''), and the percentages of the next event being of type 1 (``Type 1 follower''). 
(e)(f)(g)(h) Heatmaps of impact magnitudes between event types. Horizontal axis: source type, Vertical axis: target type. 
}
    \label{fig: real Experiment Result}
\end{figure*}

\subsubsection{Qualitative Analysis}

Our model can provide useful insights into the interaction among event types. To demonstrate this, we first quantify the magnitude of influence between event types. We compute $\int \phi_{ij}(\tau)d\tau$ for each influence function, representing the extent of influence from type $j$ (source type) to type $i$ (target type). Specifically, each learned $\int \phi_{ij}(\tau)d\tau$ is a scalar and can be demonstrated in heat maps. In this section, we analyse these datasets by looking into their learned heat maps: \cref{fig: HeatMap(SO),fig: HeatMap(Amazon),fig: HeatMap(Taobao),fig: HeatMap(Taxi)}.

\textbf{StackOverflow:} In this dataset, there are 22 event types related to ``badges'' awarded to users based on their actions. As depicted in \cref{fig: HeatMap(SO)}, many of these types have a strong positive influence on both type 4 (``Popular Question'') and type 9 (``Notable Question''). This observation aligns with the fact that ``Popular Question'' and ``Notable Question'' are the two most frequent events. Our model captures this trend and associates a significant positive impact from other types to them. 
Furthermore, a noticeable link between type 6 (``Nice Answer'', awarded when a user's answer first achieves a score of 10) and type 14 (``Enlighten'', given when a user's answer reaches a score of 10) is identified, which have nearly identical meanings. This mirrors the real-world progression from receiving a ``Nice Answer'' badge to later earning an ``Enlighten'' badge. This congruence demonstrates that our model accurately captures the dataset's characteristics and effectively highlights the interplay between different event types.

\textbf{Amazon, Taobao:} Both of these datasets pertain to customer behavior on shopping platforms and share some commonalities. Each event type represents a category of the browsing item (Taobao) or purchased item (Amazon), with Taobao having $K=16$ types and Amazon having $K=17$ types. The learned heat maps are presented in \cref{fig: HeatMap(Amazon),fig: HeatMap(Taobao)}. Interestingly, our model uncovers two common insights: 
(1) The dark diagonals observed indicate strong self-excitation for each type. This suggests customers tend to browse items of the same category consecutively in a short period. In Taobao and Amazon, with over 15 types in total, there are approximately $58.3\%$ and $21.4\%$ of events involving subsequent events of the same type. This behavior reflects how customers often browse items of the same category in a short period to decide which one to purchase. Additionally, Amazon's subscription purchases exemplify this pattern: vendors offer extra savings to customers who subscribe. These items are then regularly scheduled for delivery. 
(2) In \cref{fig: HeatMap(Amazon),fig: HeatMap(Taobao)}, rows 1 and 17 appear the darkest, indicating that these two types receive the most significant excitation from others. In reality, these two categories are the most prevalent in their respective datasets, implying that they should also have the highest intensity. What our model learns aligns empirically with the ground truth patterns in the datasets. 
Moreover, we conducted a statistical analysis on Amazon in \cref{fig: Amazon statistics}, calculating the percentages of various event types (``Total''), the percentages of the next event being of the same type (``Same type follower''), and the percentages of the next event being of type 1 (``Type 1 follower''). It is evident that the latter two constitute a significant portion ($\sim50\%$), indicating strong self-excitation effects and a pronounced exciting effect on type 1, which aligns with the learned heatmap in \cref{fig: HeatMap(Amazon)}. 

\textbf{Taxi:} In this dataset, there are 10 types of events representing taxi pick-up and drop-off across the five boroughs of New York City. Types 1-5 categorize ``drop-off'' actions, whereas types 6-10 correspond to ``pick-up'' actions in the respective boroughs. The learned heatmap (\cref{fig: HeatMap(Taxi)}) reveals three key insights:
(1) Among the ``drop-off'' actions (types 1-5), type 4 experiences the most significant influence from types 6-10 (``pick-up''). This aligns with the fact that type 4 (drop-off in Manhattan) is the most common drop-off event, accounting for over $40\%$ and thereby possessing the highest intensity.
(2) The ``pick-up'' and ``drop-off'' events always occur alternately. One driver can't pick up or drop off consecutively. As \cref{fig: HeatMap(Taxi)} shows, type 6-10 (``pick-up'') have much more excitation on type 1-5 (``drop-off'') rather than on themselves because a ``pick-up'' action will stimulate a consecutive ``drop-off'' action rather than another ``pick-up'' action. Likewise, type 1-5 have much less excitation on themselves.
(3) Type 9 and 4, pick-ups and drop-offs in Manhattan, display the most significant mutual influence, as indicated by the two darkest cells in \cref{fig: HeatMap(Taxi)}. This is consistent because most pick-up ($44.61\%$) and drop-off ($42.89\%$) actions occur in Manhattan. Furthermore, these two types always occur in tandem: $90.8\%$ of passengers picked up in Manhattan are also dropped off there, and $96.2\%$ of drivers who complete a trip in Manhattan will pick up their next customer within the same borough. This behavior is a clear short-term pattern captured by our model and is evident in the dataset. 
\subsection{Ablation Study}
\label{section: ablation study}
Our model has reduced the parameter count but still achieves comparable or even better results compared to THP. This improvement is attributed to the ``fully attention-based intensity function'' (\cref{section: Fully Attention-based Intensity Function}). THP relies on the parameterized extrapolated intensity, assuming that the intensity function on non-event intervals follows an approximately linear pattern (red term in \cref{thp_lamda}). However, such an assumption does not align with the actual patterns in real data and can impact the expressive capability of the model. 
We conduct further ablation studies to illustrate the limitations of the parameterized extrapolation method in \cref{table: Real Data Experiment}. 
We implement an extra revised model \textbf{Ex-ITHP} which essentially is ``interpretable Transformer'' + ``extrapolated intensity''. More details about Ex-ITHP is provided in \cref{appendix: section: implementation of Ex-ITHP}. 
THP, ITHP, and Ex-ITHP naturally constitute an ablation study. 
Ex-ITHP has fewer parameters as it removes the parameters $\mathbf{W}^Q$ and $\mathbf{W}^K$, and uses a less flexible extrapolated intensity. 
In \cref{table: Real Data Experiment}, the Ex-ITHP exhibits the poorest performance due to its fewer parameters and restricted intensity flexibility. THP performs moderately, having more parameters but still restricted intensity flexibility. Conversely, the ITHP, despite having fewer parameters, outperforms THP on most datasets owing to its more flexible intensity expression. 
\begin{figure}[t]
    \centering
    \includegraphics[width=0.8\columnwidth]{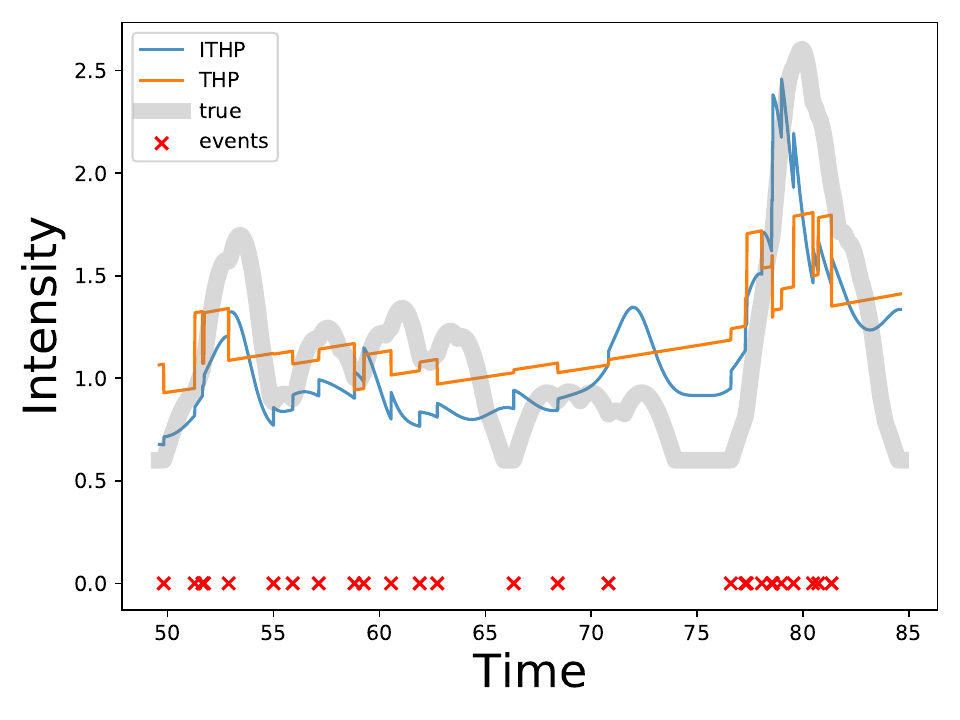}
    \vspace*{-3mm}
    \caption{Comparison of the learned intensity for a segment of a sequence by ITHP and THP. 
    THP fails to learn the fluctuating intensity on non-event intervals, but only maintains an approximately linear pattern due to the extrapolation assumption. 
    In contrast, our proposed ITHP demonstrates greater flexibility, successfully capturing the fluctuating intensity, and accurately fitting the scale level.}
    \label{fig: comparison between extrapolation and fully-attention}
    \vspace*{-5mm}
\end{figure}
Additionally, we visualize the difference between the learned fully attention-based intensity and the learned extrapolated intensity for a segment of the sequence in Half Sinusoidal Kernel Hawkes Synthetic dataset (\cref{section: Toy Data}). As depicted in \cref{fig: comparison between extrapolation and fully-attention}, 
on non-event intervals, THP, constrained by the approximately linear extrapolation, struggles to capture the fluctuating intensity patterns and can only learn an intensity that is approximately linear. Additionally, due to the limited variation in intensity on non-event intervals, large jumps are required when a new event occurs to maintain a height similar to the ground-truth intensity. 
In contrast, our proposed ITHP demonstrates greater flexibility, successfully capturing the fluctuating pattern on non-event intervals, and accurately fitting the scale level. 
\subsection{Hyperparameter Analysis}
\label{hyperparameter analysis}

Our model's configuration primarily encompasses two dimensions: the encoding dimension, denoted as $M$, and the Value dimension, denoted as $M_V$. 
We maintain the skip connection within the implementation of the encoder which necessitates that $M_V$ must be equal to $2M$.
We test the sensitivity of model performance to hyperparameters by using various hyperparameter configurations on one toy dataset and one public dataset: the half-sine and Taxi datasets. The results of our experiments are shown in~\cref{table: Hyper-Parameter Configurations}. The results indicate that our model is not significantly affected by the hyperparameter variation. Additionally, it can achieve reasonably good performance even with fewer parameters.

\begin{table}[h]
    \centering
    \caption{Experiments of different hyperparameters configurations on datasets: half-sine and Taxi. The results indicate that our model is robust to these hyperparameters.}
    \begin{sc}
    \begin{tabular}{ccccc}
        \toprule
        \multirow{2}{*}{Config} & \multicolumn{2}{c}{Taxi} & \multicolumn{2}{c}{Half-Sine}\\
         & TLL & ACC & TLL & ACC\\
        \midrule
        $M=64$, $M_V=128$ & 0.2513 & 0.97 & -0.7714 & 0.58\\
        \midrule
        $M=128$, $M_V=256$ & 0.2501 & 0.97 & -0.7909 & 0.58\\
        \midrule
        $M=256$, $M_V=512$ & 0.2520 & 0.97 & -0.7822& 0.58\\
        \midrule
        $M=512$, $M_V=1024$ & 0.2498 & 0.97 & -0.7852 & 0.59\\
        \bottomrule
    \end{tabular}
    \end{sc}
    \label{table: Hyper-Parameter Configurations}
\end{table}

\section{Conclusion}

To model interactions in social networks using event sequence data, we introduce ITHP as a novel approach to enhance the interpretability and expressive power of deep point processes model. 
Specifically, ITHP not only inherits the strengths of Transformer Hawkes processes but also aligns with statistical nonlinear Hawkes processes, offering practical insights into user or group interactions. It further enhances the flexibility of intensity functions over non-event intervals. 
Our experiments have demonstrated the effectiveness of ITHP in overcoming inherent limitations in existing deep point process models. Our findings open new avenues for research in understanding and modeling the complex dynamics of social ecosystems, ultimately contributing to the broader understanding of these intricate networks.

\section*{Acknowledgments}
This work was supported by NSFC Project (No. 62106121), the MOE Project of Key Research Institute of Humanities and Social Sciences (22JJD110001), and the Public Computing Cloud, Renmin University of China. 
\bibliographystyle{ACM-Reference-Format}
\bibliography{sample-base}
\appendix

\section*{Appendix}

\section{Toy Data}
\label{Appendix: Synthetic Data}
We simulate two synthetic datasets: the exponential decay Hawkes processes and the half sinusoidal Hawkes processes. In the case of the exponential decay Hawkes processes, we set the maximum observed length as $T=20$. For the half sinusoidal Hawkes processes, the maximum observed length is set to $T=100$. 
To perform simulation, we employ the thinning algorithm~\citep{ogata1998space} which is outlined in the algorithm below. Note that the kernel function $\phi_{mn}(\tau)$ is defined in~\cref{section: Toy Data}, where $m$ is the target type and $n$ is the source type. Note that $\phi_{m*}(\tau)$ indicates all kernels whose target type is $m$. The statistics of our toy datasets are listed in~\cref{table: toy data statistics}. 
Additionally, we present the attention weight matrix of a testing sequence in the half-sine toy data, as depicted in~\cref{fig: Attention Map(Half-Sine)}.

\begin{algorithm}
\caption{Simulation of an $M$-variate Hawkes processes with kernels $\phi_{mn}(\cdot)$ for $m,n = 1,2,\dots,M$ on $[0,T]$.}
\begin{algorithmic}
    \Require{\{$\mu_n$, $\phi_{mn}(\cdot)$\} for $m,n = 1,2,\dots,M$, the observation window $[0,T]$}
    \State Initialize $\mathcal{T}^1 = \dots = \mathcal{T}^M = \emptyset$, $n^1 = \dots = n^M = 0$, $s = 0$;
    \While{$s < T$}
        \State Set $\bar{\lambda} = \sum_{m=1}^M \lambda^m(s^{-}) + \text{max}(\phi_{m*}(\cdot))$;
        \State Generate $w \sim \text{exponential}(1/\bar{\lambda})$; 
        \State Set $s = s + w$; 
        \State Generate $D \sim \text{uniform}(0,1)$;
        \If{$D\bar{\lambda} \leq \sum_{m=1}^M \lambda^m(s) $}
            \State $k\sim\text{categorical}([\lambda^1(s),\ldots,\lambda^M(s)]/\bar{\lambda})$;
            \State $n^k = n^k + 1$; 
            \State $t^k_{n^k} = s$
            \State $\mathcal{T}^k=\mathcal{T}^k \cup \{t^k_{n^k}\}$
        \EndIf
    \EndWhile
    \If{$t^k_{n^k} \leq T$}
        \State \Return $\{ \mathcal{T}^m \}$ for $m = 1,2,\dots,M$
    \Else
        \State \Return $\mathcal{T}^1 \dots \mathcal{T}^k \slash \{t^k_{n^k}\} \dots \mathcal{T}^M$
    \EndIf 
\end{algorithmic}
\end{algorithm}

\begin{table*}[h]
    \centering
    \caption{The statistics of two synthetic datasets. Details such as the number of events, the sequence length and the event intervals are provided below.}
    \begin{sc}
    \scalebox{1.0}{
    \begin{tabular}{ccccccccc}
        \toprule
        \multirow{2}{*}{Dataset} & \multirow{2}{*}{Split} & \multirow{2}{*}{\# of Events} & \multicolumn{3}{c}{Sequence Length} & \multicolumn{3}{c}{Event Interval}\\
            &   &   & Max & Min & Mean(Std) & Max & Min &Mean(Std)\\
        \midrule
        \multirow{3}{*}{Exponential-Decay} & training & 70644 & 877 & 47 & 282.58 (150.25) & 21.86 & 1.91e-06 & 0.28 (150.25)\\
            & validation & 36521 & 877 & 47 & 292.17 (153.94) & 20.90 & 1.91e-06 & 0.27 (153.94)\\
            & test & 33716 & 894 & 59 & 269.73 (144.97) & 20.31 & 3.81e-06 & 0.29 (144.97)\\
        \midrule
        \multirow{3}{*}{Half-Sine} & training &95714 & 858 & 158 & 382.86 (96.29) & 21.21 & 3.83e-07 & 0.52 (96.29)\\
            & validation &48814 & 858 & 158 & 390.51 (106.12) & 21.21 & 3.83e-07 & 0.51 (106.12) \\
            & test& 50376 & 717 & 223 & 403.01 (102.19) & 19.57 & 5.48e-06 & 0.49 (102.19)\\
        \bottomrule
    \end{tabular}}
    \end{sc}
    \label{table: toy data statistics}
\end{table*}

\section{Public Data}
\label{Appendix: Public Data}
\subsection{Public Data Statistics}
In this section, we cover the main statistics of five public datasets: StackOverflow, Taobao, Amazon, Taxi and Conttime, which is listed in~\cref{table: Public statistics}. 
Note that all the public data are multivariate. The visualizations of the event percentages in each dataset are depicted in~\cref{fig: public statistics}. Each subplot in~\cref{fig: public statistics} displays the distribution of event types in the training, validation, and testing sets, respectively. 
\begin{table*}[ht]
    \centering
    \caption{The statistics of five public datasets. Details such as the number of events, the statistics of event interval are provided.}
    \begin{sc}
    \scalebox{1.0}{
    \begin{tabular}{ccccccccc}
        \toprule
        \multirow{2}{*}{Dataset} & \multirow{2}{*}{Split} & \multirow{2}{*}{\# of Events} & \multicolumn{3}{c}{Sequence Length} & \multicolumn{3}{c}{Event Interval}\\
            &   &   & Max & Min & Mean(Std) & Max & Min &Mean(Std)\\
        \midrule
        \multirow{3}{*}{STACKOVERFLOW} & training &90497 & 101 & 41 & 64.59 (20.46) & 20.34 & 1.22e-4 & 0.88 (20.46) \\
            & validation & 25313 & 101 & 41 & 63.12 (19.85) & 16.68 & 1.22e-4 & 0.90 (19.85)\\
            & test &26518 & 101 & 41 & 66.13(20.77) & 17.13 & 1.22e-4 & 0.85 (20.77) \\
        \midrule
        \multirow{3}{*}{TAOBAO} & training & 75205 & 64 & 40 & 57.85 (6.64) & 2.00 & 9.99e-05 & 0.22 (6.64)\\
            & validation &11497 & 64 & 40 & 57.49 (6.82) & 1.99 & 9.99e-05 & 0.22 (6.82) \\
            & test & 28455 & 64 & 32 & 56.91 (7.82) & 1.00 & 4.21e-06 & 0.05 (7.82)\\
        \midrule
                \multirow{3}{*}{AMAZON} & training &288377 & 94 & 14 & 44.68 (17.88) & 0.80 & 0.010 & 0.51 (17.88) \\
            & validation & 40088 & 94 & 15 & 43.48 (16.60) & 0.80 & 0.010 & 0.50 (16.60)\\
            & test & 84048 & 94 & 14 & 45.41 (18.19) & 0.80 & 0.010 & 0.51(18.19)\\
        \midrule
                \multirow{3}{*}{TAXI} & training & 51854 & 38 & 36 & 37.04 (1.00) & 5.72 & 2.78e-4 & 0.22 (1.00)\\
            & validation & 7422 & 38 & 36 & 37.11 (1.00) & 5.52 & 2.78e-4 & 0.22 (1.00)\\
            & test & 14820 & 38 & 36 & 37.05 (1.00) & 5.25 & 8.33e-4 & 0.22 (1.00)\\
        \midrule
                \multirow{3}{*}{CONTTIME} & training & 479467 & 100 & 20 & 59.93 (23.13) & 4.03 & 1.91e-06 & 0.24 (23.13)\\
            & validation & 60141 & 100 & 20 & 60.14 (22.97) & 3.94 & 2.86e-06 & 0.24 (22.97)\\
            & test & 61781 & 100 & 20 & 61.78 (23.21) & 4.47 & 9.54e-07 & 0.24 (23.21)\\
        \bottomrule
    \end{tabular}}
    \end{sc}
     \label{table: Public statistics}
\end{table*}

\begin{figure*}
    \begin{center}
    \adjustbox{valign=b}{
    \begin{minipage}{0.26\textwidth}
    \includegraphics[width=\columnwidth]{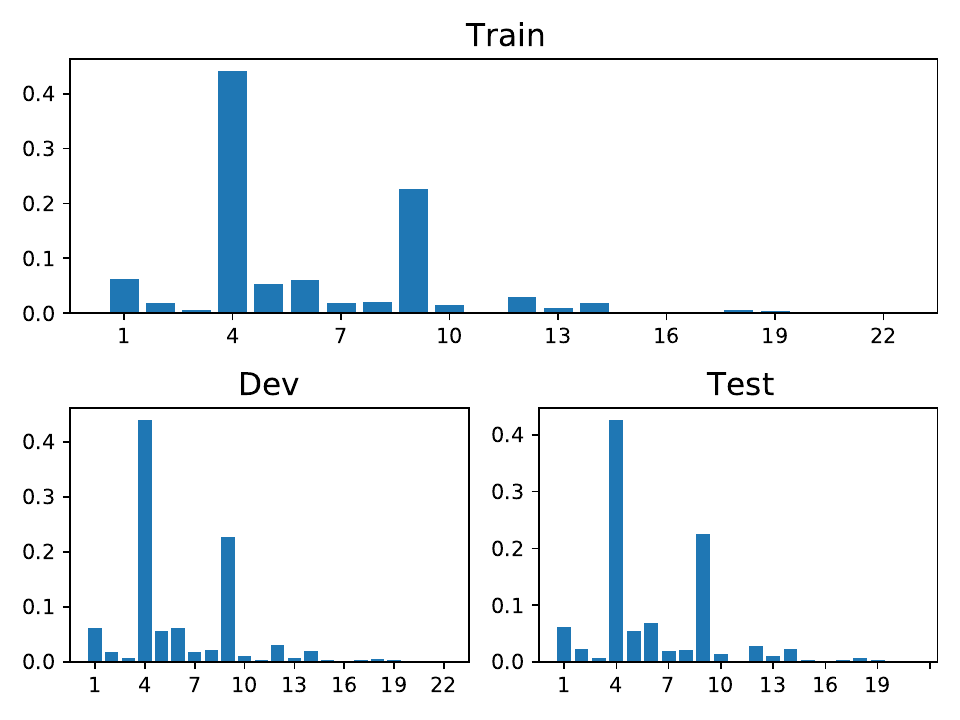}
    \subcaption[]{StackOverflow}
    \label{fig: StackOverflow Type Portion}
    \end{minipage}}
    \adjustbox{valign=b}{
    \begin{minipage}{0.26\textwidth}
    \includegraphics[width=\columnwidth]{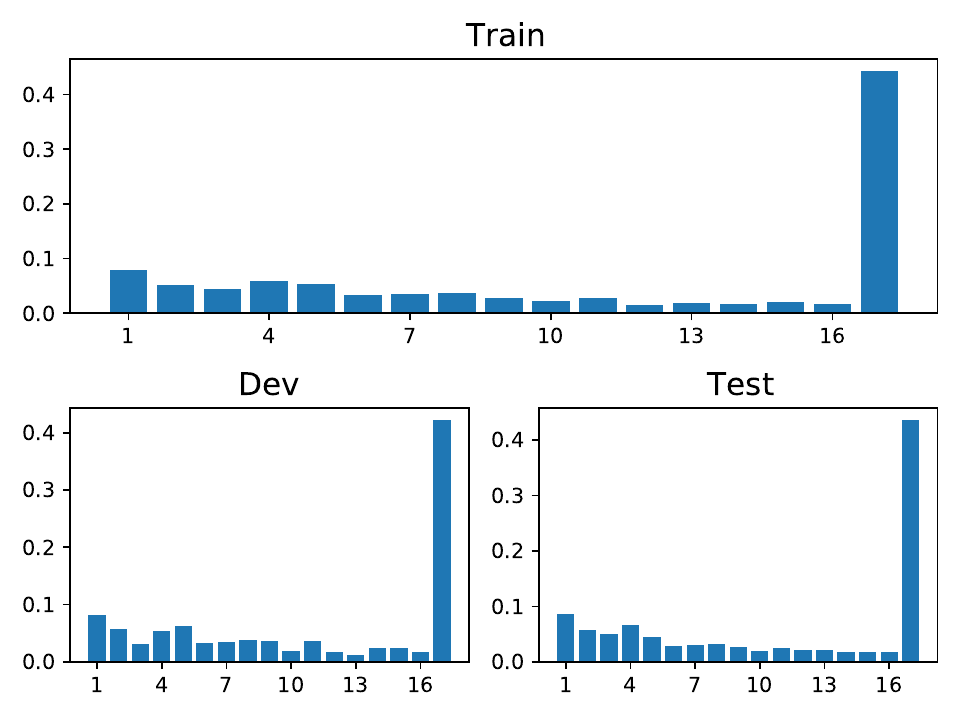}
    \subcaption[]{Taobao}
    \label{fig: Taobao Type Portion}
    \end{minipage}}
    \adjustbox{valign=b}{
    \begin{minipage}{0.26\textwidth}
    \includegraphics[width=\columnwidth]{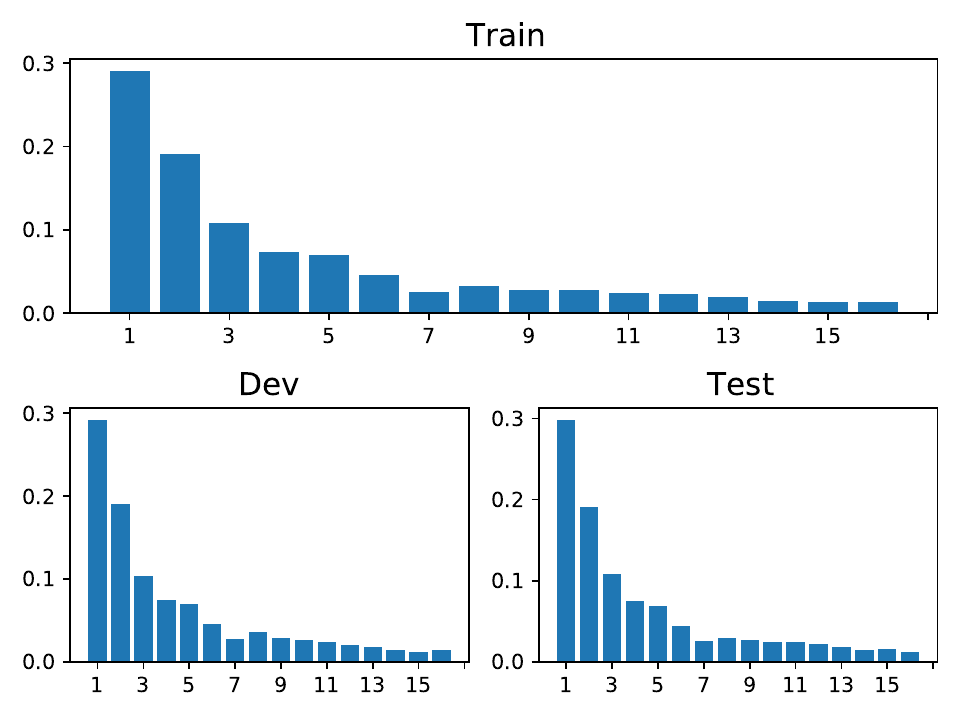}
    \subcaption[]{Amazon}
    \label{fig:amazon Type Portion}
    \end{minipage}}
    \adjustbox{valign=b}{
    \begin{minipage}{0.26\textwidth}
    \includegraphics[width=\columnwidth]{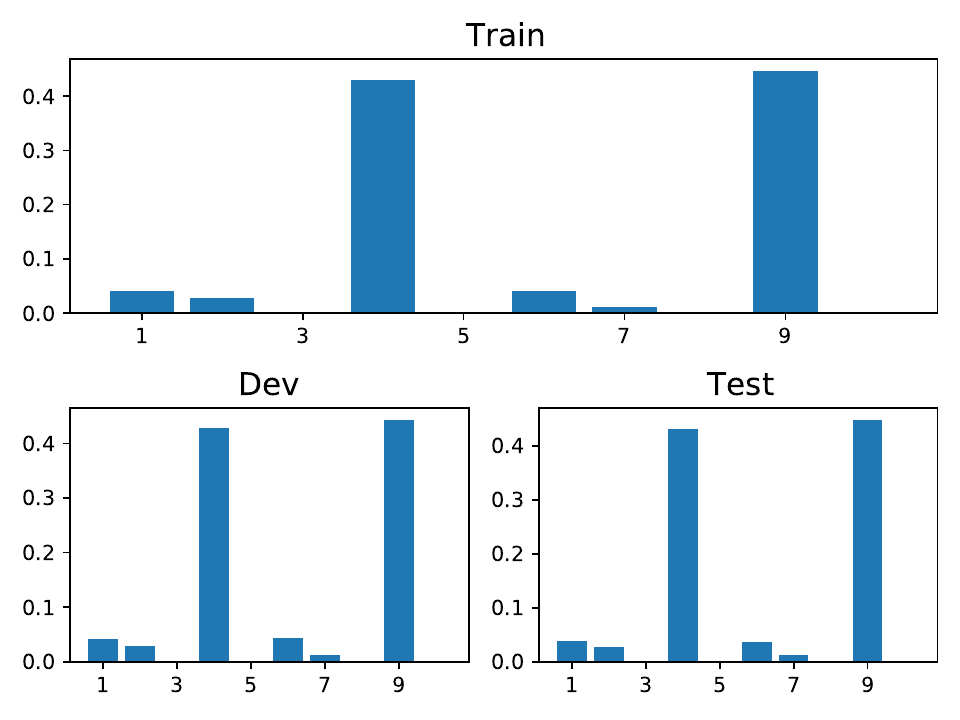}
    \subcaption[]{Taxi}
    \label{fig: Taxi Type Portion}
    \end{minipage}}
    \adjustbox{valign=b}{
    \begin{minipage}{0.26\textwidth}
    \includegraphics[width=\columnwidth]{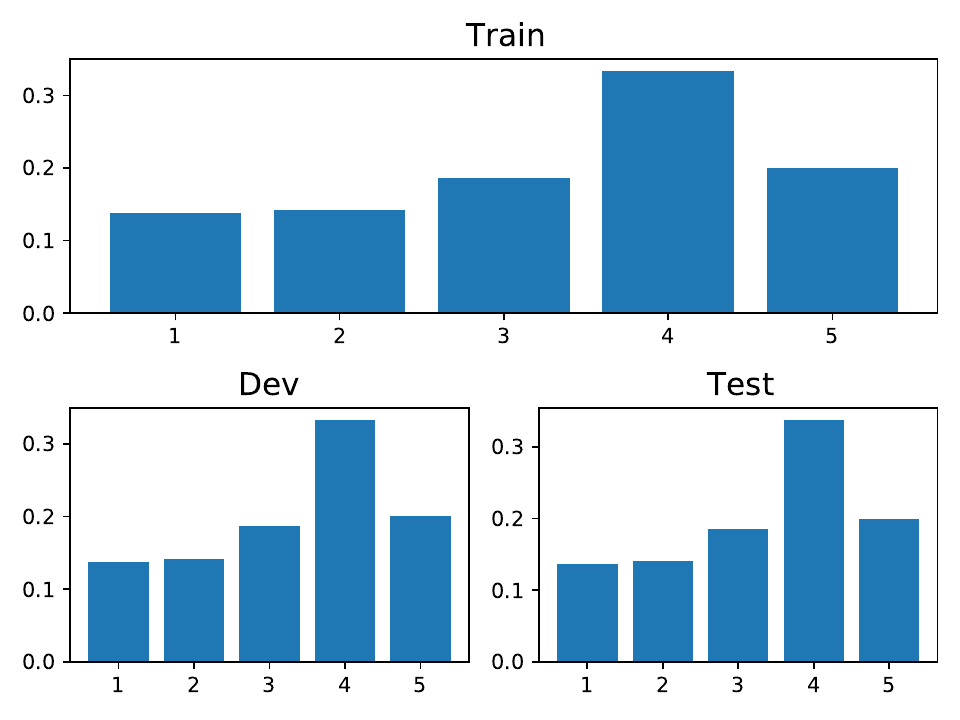}
    \subcaption[]{Conttime}
    \label{fig:conttime portion}
    \end{minipage}}
    \end{center}
    \caption{All the public datasets consist of multiple event types. Specifically, Amazon has $K=17$ event types, StackOverflow has $K=22$, Taxi has $K=10$, Taobao has $K=16$, and Conttime has $K=5$. We provide visualizations of the event percentages in each dataset. Each subplot illustrates the distribution of event types in the training, validation, and testing sets, respectively. Notably, there is a significant imbalance in event types observed in StackOverflow, Taxi, and Taobao datasets.}
    \label{fig: public statistics}
\end{figure*}

\subsection{Additional Attention Map}
We present additional attention weight matrices for the four public datasets, as depicted in~\cref{fig: Additional Attention Map}, which is more intuitive to demonstrate how events affect each other in the sequence. 
\begin{figure*}
    \begin{center}
    \adjustbox{valign=b}{
    \begin{minipage}{0.26\textwidth}
    \includegraphics[width=\columnwidth]{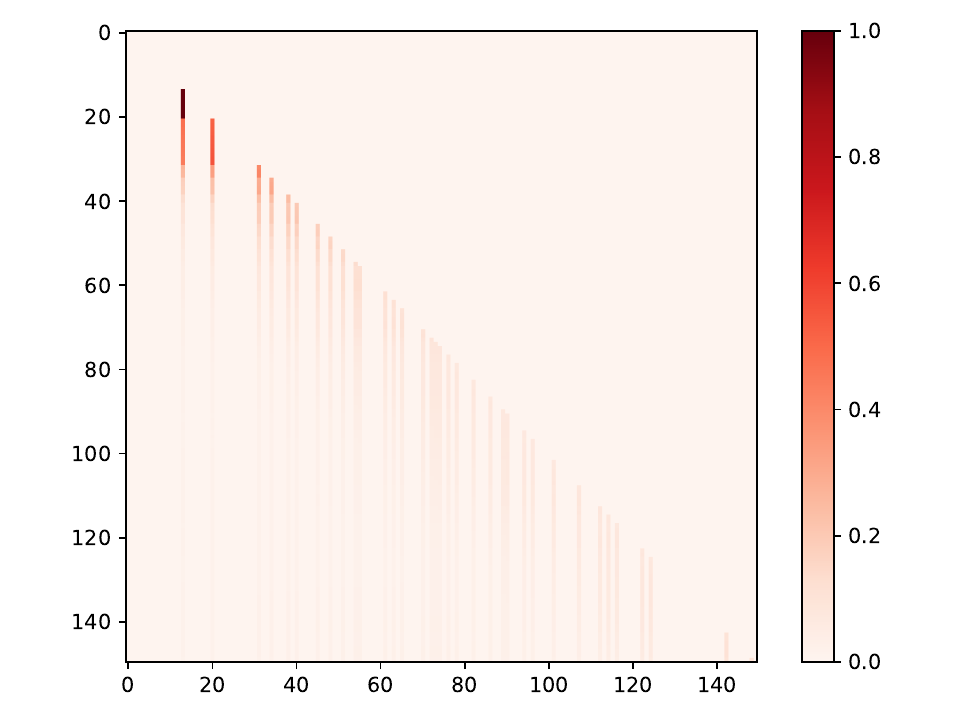}
    \subcaption[]{Half-sine}
    \label{fig: Attention Map(Half-Sine)}
    \end{minipage}}
    \adjustbox{valign=b}{
    \begin{minipage}{0.26\textwidth}
    \includegraphics[width=\columnwidth]{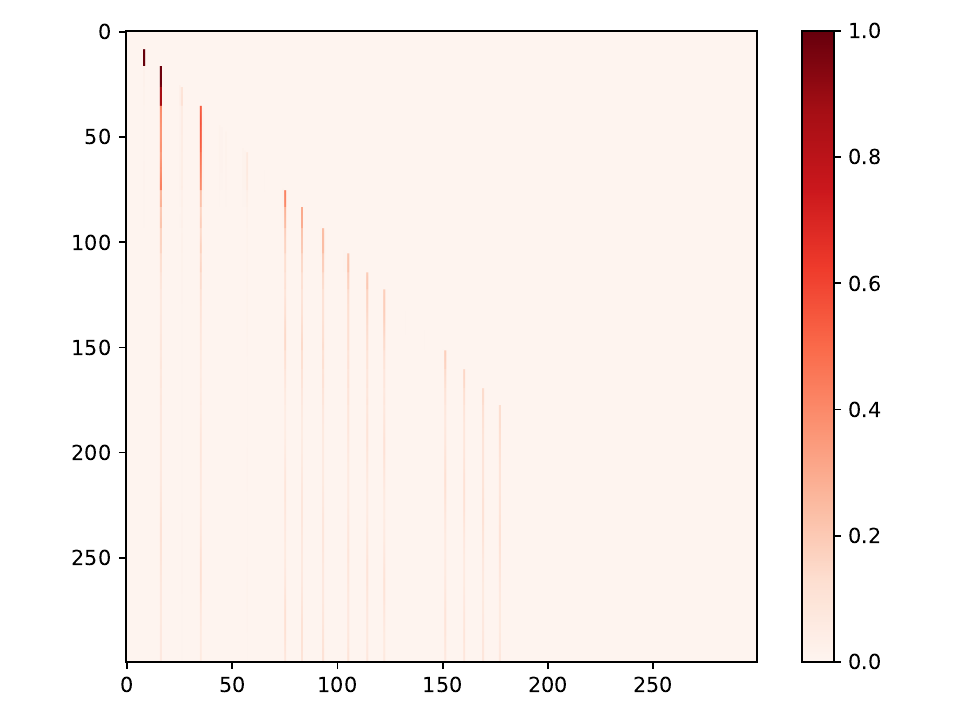}
    \subcaption[]{Taobao}
    \label{fig: Attention Map(Taobao)}
    \end{minipage}}
    \adjustbox{valign=b}{
    \begin{minipage}{0.26\textwidth}
    \includegraphics[width=\columnwidth]{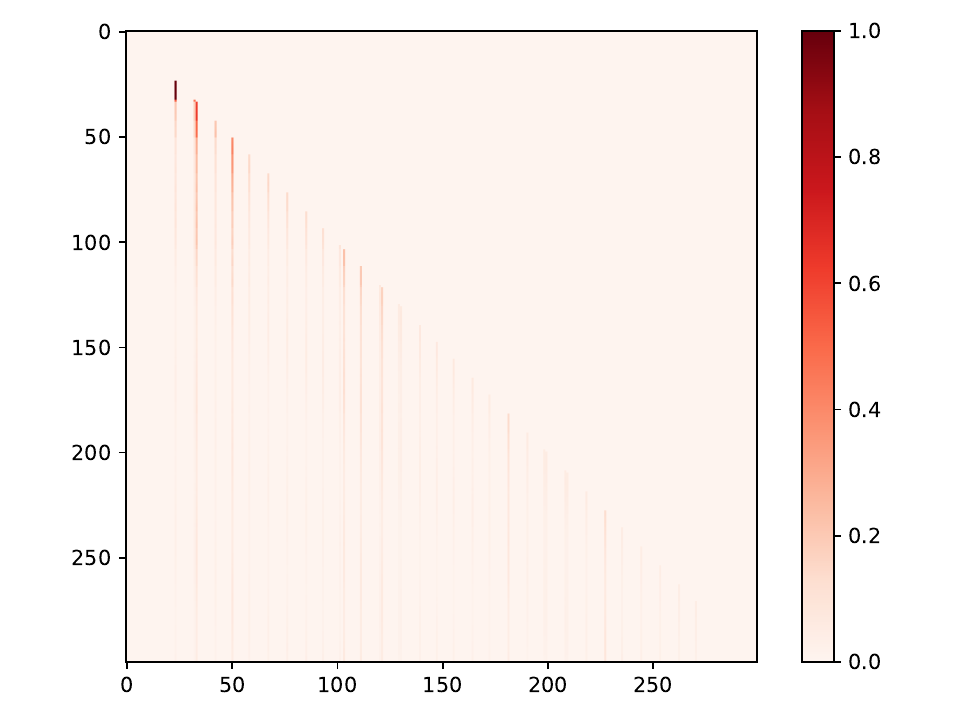}
    \subcaption[]{Amazon}
    \label{fig: Attention Map(AMAZON)}
    \end{minipage}}
    \adjustbox{valign=b}{
    \begin{minipage}{0.26\textwidth}
    \includegraphics[width=\columnwidth]{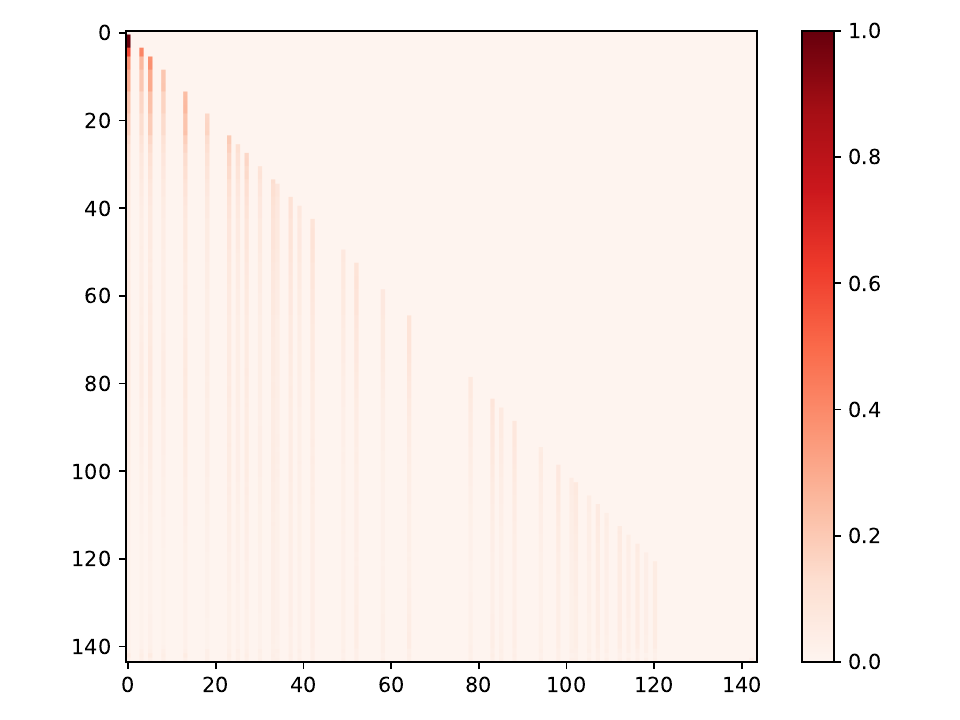}
    \subcaption[]{Taxi}
    \label{fig: Attention Map(Taxi)}
    \end{minipage}}
    \adjustbox{valign=b}{
    \begin{minipage}{0.26\textwidth}
    \includegraphics[width=\columnwidth]{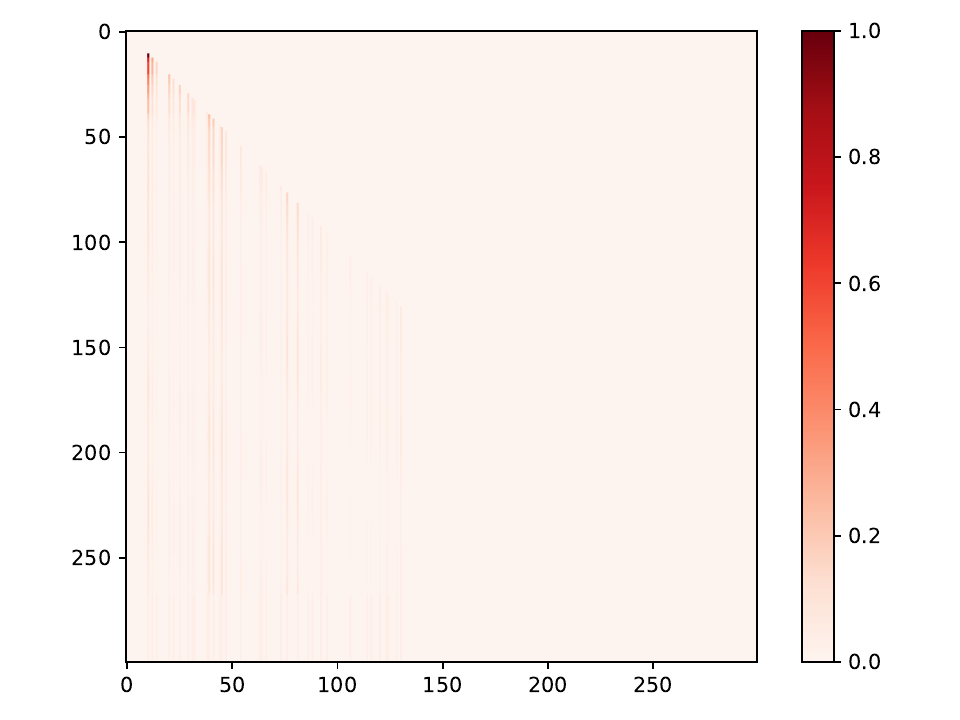}
    \subcaption[]{Conttime}
    \label{fig: Attention Map(CONTTIME)}
    \end{minipage}}
    \end{center}
    \caption{Each subplot demonstrates the attention map of a testing sequence in the half-sine toy data and the other four public datasets. Similarly, the horizontal axis represents the source point, while the vertical axis represents the target point. It is apparent that events have a lasting effect on subsequent events. Grids within non-event intervals do not exert any future influence, as they do not correspond to actual event occurrences. Each vertical rectangle signifies a single event's impact on various future events. The color fading from top to bottom within each rectangle indicates that the impact diminishes over time.}
    \label{fig: Additional Attention Map}
\end{figure*}

\section{Implementation of Ex-ITHP}
\label{appendix: section: implementation of Ex-ITHP}
The Ex-ITHP, namely, "extrapolation iTHP" is the iTHP(removing parameters $W^Q$, $W^K$) utilising 'extrapolation intensity'. In this section, we introduce the implementation of Ex-ITHP.
Given a sequence $\mathcal{S} = \{(t_i, k_i)\}_{i=1}^{L}$ where each event is characterized by a timestamp $t_i$ and an event type $k_i$, Ex-ITHP utilize the same temporal embedding and type embedding and concatenates them as iTHP does:

\begin{equation}
    \mathbf{X}=[\mathbf{Z},\mathbf{E}]\in\mathbb{R}^{L \times 2M}, 
\end{equation}

where $\mathbf{Z} \in \mathbb{R}^{L \times M}$ and $\mathbf{E} \in \mathbb{R}^{L \times M}$ are the temporal encoding and type encoding of $\mathcal{S}$. The encoder output $\mathbf{S}$ is calculated in the same way as~\cref{new_transformer_attention}. $\mathbf{S}_i$ indicates the $i$-th row of $\mathbf{S}$, which is the representation of event $i$:
\begin{equation}
    \mathbf{S}_i = \sum_{j<i}\underset{j<i}{\text{softmax}}\left(\frac{\mathbf{X}_i\mathbf{X}_j^\top}{\sqrt{2M}}\right)\mathbf{V}_j \in \mathbb{R}^{M_V}. 
\end{equation}
The encoder output $\mathbf{S}_i$ is then passed through a MLP to get the final representation of event $i$:
\begin{equation*}
\mathbf{H}_i = \text{ReLU}(\mathbf{S}_i\mathbf{W}_1 + \mathbf{b}_1)\mathbf{W}_2 + \mathbf{b}_2 \in \mathbb{R}^{M}.
\end{equation*}
Then, given a type $k$ and time t ($t_i<t \leq t_{i+1}$), the corresponding intensity is given by the extrapolation method:
\begin{equation}
\lambda_k(t|\mathcal{H}_t)=\text{softplus} (\alpha_k \frac{t-t_i}{t_i}+\mathbf{w}_k^\top\mathbf{H}_i+b_k).
\end{equation}
Finally, the log-likelihood to be optimized is given by:

\begin{equation}
\mathcal{L}(\mathcal{S})=\sum_{i=1}^L\log\lambda_{k_i}(t_i|\mathcal{H}_{t_i})-\sum_{k=1}^K\int_0^T\lambda_k(t|\mathcal{H}_t)dt, 
\end{equation}

In summary, Ex-ITHP employs identical encoding techniques, specifically temporal encoding and event type encoding, while also eliminating the use of $W^Q$ and $W^K$, akin to iTHP. However, it utilizes the extrapolation method to formulate the intensity as THP does.

\end{document}